\documentclass[conference]{IEEEtran}

\usepackage[usenames,dvipsnames]{color}
\usepackage{graphicx}
\usepackage{amssymb}
\usepackage{amsmath}
\usepackage{amsthm}
\usepackage{verbatim}
\usepackage{url}
\usepackage{tabularx}
\usepackage{eurosym}
\usepackage{type1cm}
\usepackage{eso-pic}
\usepackage{color}
\usepackage{algorithm}
\usepackage{algorithmic}
\newcommand{\squishlist}{
 \begin{list}{$\bullet$}
  { \setlength{\itemsep}{0pt}
     \setlength{\parsep}{3pt}
     \setlength{\topsep}{3pt}
     \setlength{\partopsep}{0pt}
     \setlength{\leftmargin}{1.5em}
     \setlength{\labelwidth}{1em}
     \setlength{\labelsep}{0.5em} } }

\newcommand{\squishend}{
  \end{list}  }
\newcommand{\I}{\operatorname{I}}
\newcommand{\triplets}{\mathcal{T}}

\newcommand{\opt}{{\operatorname{OPT}}}
\newcommand{\vc}{\operatorname{vc}}
\newcommand{\alg}{\operatorname{ALG}}
\newcommand{\knew}{k_{\operatorname{new}}}
\newcommand{\vcapprox}{\textsf{Approx-VC}}

\newtheorem{problem}{Problem}
\newtheorem{definition}{Definition}
\newtheorem{theorem}{Theorem}
\newtheorem{corollary}{Corollary}
\newtheorem{lemma}{Lemma}
\newtheorem{assumption}{Assumption}

\begin{document}

\title{Crowdsourced correlation clustering with relative distance comparisons}
\author{\IEEEauthorblockN{Antti Ukkonen}
\IEEEauthorblockA{Department of Computer Science, University of Helsinki\\
Helsinki, Finland\\
Email: antti.ukkonen@helsinki.fi}}
\maketitle
\begin{abstract}
Crowdsourced, or human computation based clustering algorithms usually
rely on relative distance comparisons, as these are easier to elicit
from human workers than absolute distance information. A relative
distance comparison is a statement of the form ``item A is closer to
item B than to item C''. However, many existing clustering algorithms
that use relative distances are rather complex. They are often based
on a two-step approach, where the relative distances are first used to
learn either a distance matrix, or an embedding of the items, and then
some standard clustering method is applied in a second step. In this
paper we argue that it should be possible to compute a clustering
directly from relative distance comparisons.

Our ideas are built upon existing work on correlation clustering, a
well-known non-parametric approach to clustering. The technical
contribution of this work is twofold. We first define a novel variant
of correlation clustering that is based on relative distance
comparisons, and hence suitable for human computation. We go on to
show that our new problem is closely related to basic correlation
clustering, and use this property to design an approximation algorithm
for our problem. As a second contribution, we propose a more
practical algorithm, which we empirically compare against existing
methods from literature. Experiments with synthetic data suggest that
our approach can outperform more complex methods. Also, our method
efficiently finds good and intuitive clusterings from real relative
distance comparison data.

\end{abstract}

\section{Introduction}
Clustering is a classical unsupervised learning problem.
The task, in colloquial terms,
is to divide a given set of items
to groups such that
similar items are placed in the same group,
while dissimilar items end up in different groups.
Clustering has numerous practical applications,
ranging from customer segmentation to bioinformatics,
and has attracted a lot of attention from
the research community for decades \cite{XuW05}.

In this paper we
study a novel approach to data clustering
that is suitable for {\em human computation} \cite{law2011human,QuinnB11},
i.e., an algorithmic process where
humans carry out parts of the computation,
often using a crowdsourcing platform such as Amazon Mechanical Turk.
Human computation algorithms
are implemented 
via so called human intelligence tasks (HITs, see also \cite{Ipeirotis10}).
A HIT is defined as a piece of input data together with instructions of what to do with the input.
Human computation algorithms operate by sending a large number of HITs to a crowd that processes the tasks in parallel. Once all tasks are completed, the algorithm collects the results and possibly carries out some post-processing to obtain the final result.

To motivate
human computation approaches to clustering,
consider a scenario
where we are given a collection of some items,
e.g.~photographs or pieces of text,
and ask human labellers to assign
each item to some category.
Despite recent advances in computer vision (e.g.~\cite{RussakovskyDSKS15}),
the need for this type of crowdsourced data analysis remains in scenarios
where human performance still exceeds that of machine learning.
In simple cases the categories (or labels) of interest are known
e.g., they can correspond to images different types of galaxies \cite{lintott2008galaxy},
or to texts having positive / negative sentiment \cite{mohammad2013crowdsourcing}.
But in some other situations,
there may not be any predefined categories,
and the first task is to
determine what kind of structure there is in the data to begin with.
This is a data exploration problem
to which clustering is a standard solution.

To design an efficient human computation algorithm
we should make sure that the required HITs
a) {\em are easy for humans to solve}, and
b) {\em can all be solved in parallel}.
Next we discuss
the technical motivation of our approach
on the basis of these two requirements.

At the core of most clustering algorithms
is the notion of {\em distance}.
Indeed, similarity between two items is usually defined
in terms of a distance function that
yields a small numeric value when the items are similar,
and increases as the items become more dissimilar.
A lot of the actual computation carried out by a clustering algorithm
involves calculating, comparing or otherwise using these distances.
Any human computation algorithm for clustering
must thus deal with distances as well,
and it must do this in a manner
that satisfies both requirements a) and b) above.

The main problem is that
absolute (numeric) distances between e.g.~images
can be difficult for humans to specify in a consistent manner.
Even a very simple distance function 
that can take only two possible values,
``similar'' and ``not-similar'',
can be problematic for human annotators \cite{PeiFTR16}.
{\em Relative distance comparisons},
on the other hand,
are often easier to elicit.
Rather than specifying distances on some absolute (and arbitrary) scale,
they represent the distance function
in terms of statements such as
{\em ``item A is closer to item B than to item C''},
or
{\em ``of items A, B, and C, item C is an outlier''}.
Relative distance comparisons of this kind
have been used previously
e.g.~to compute the mean of a set of items \cite{HeikinheimoU13},
density estimation \cite{UkkonenDH15,KleindessnerL16},
distance/kernel learning \cite{SchultzJ03,GomesWKP11,TamuzLBSK11,AmidGU15,KleindessnerL16a},
and to compute embeddings \cite{MaatenW12,AmidU15,AmidVW16}.
To satisfy requirement a) above,
the clustering algorithm must thus use relative distance comparisons only.
Requirement b) is satisfied as long as
the distance comparisons can be collected in one batch,
i.e., there are no interdependencies between the HITs.

% Note that a setup somewhat similar to correlation clustering
% is also the basis of several

% in the sense that
% it is possible to e.g.~compute sums of (squared) distances,
% or use the distance as an argument to a probability density function,
% when computing the quality of likelihood of a clustering solution.

% However, unlike in semi-supervised clustering,
% in correlation clustering {\em no other information (e.g.~item features)
% is available.}
Most existing
human computation algorithms for clustering
that satisfy requirements a) and b)
first use the relative distance comparisons
to either learn a distance/kernel matrix \cite{GomesWKP11,TamuzLBSK11}, or
to compute an embedding of the items to $\mathbb{R}^n$ \cite{MaatenW12},
and as a second step apply some ``regular'' clustering algorithm
that uses distance matrices or embeddings.
Such approaches indeed can work well,
and have the sometimes useful property of
learning features for the items (the embedding).

But we argue that clustering is
a relatively simple combinatorial problem.
If the ultimate task is to only compute a clustering,
and there is no other use for e.g.~an embedding of the items,
%% instead of first producing some intermediary representation
it seems more desirable to
{\em compute the clustering directly
from the distance comparisons}.
The aim of this paper is to devise
an efficient method for doing this.
%% In addition to the technical requirements a) and b) above,
We argue that our approach
is 1) conceptually simpler than competing approaches,
and 2) easier to implement and understand.

% Also note that we are not concerned with
% semi-supervised clustering methods that
% combine item features with
% e.g.~{\em must-link} and {\em cannot-link} constraints \cite{WagstaffC00}.
% While such methods can be useful in many situations,
% in this paper we make the additional restriction
% that features are not available.

% The variant of $k$-means proposed in \cite{HeikinheimoU13}
% satisfies requirement a),
% but fails to satisfy requirement b),
% as the algorithm is relatively complex, and requires
% processing the HITs in several batches
% with intermediary computations in between.

% However,
% most clustering algorithms define item similarity
% Indeed,
% the well-known k-means clustering algorithm \cite{lloyd1982least}
% aims to minimise the sum of squared distances
% of every data item to its cluster centroid.
% Or, the likelihood of a Gaussian mixture model
% (see e.g.~Section 6.8 in \cite{HastieTF09})
% is computed in terms of absolute distances as well.

% However, human annotators may find that
% relative distance comparisons
% are easier to specify
% than must-link / cannot-link constraints ,
% let alone absolute distances.

% Correlation clustering seems like a promising framework
% to implement an algorithm that addresses these shortcomings.
% The main objective of this paper
% is thus to address this question.

% Another property of many clustering algorithms is
% that they require the user to specify the resulting number of clusters
% as a parameter.
% This is true for k-means, as well as for many probabilistic methods as well.
% Bla bla...

Our approach is based on
\textsc{Correlation-Clustering},
a problem
originally proposed by Bansal et al \cite{BansalBC04}\footnote{Note that B{\"o}hm et al \cite{BohmKKZ04} use the term ``correlation clustering'' to describe a {\em different problem} that is not to be confused with the one considered in this paper.}.
It is a parameter-free approach to clustering,
where the distance function and item features have been replaced with
``qualitative'' information about how pairs of items
should relate to each other in a good clustering.
A lot of its theoretical properties are known \cite{DemaineEFI06,AilonCN08},
and unlike many other clustering methods,
it does not require setting the number of clusters in advance,
which is a nice practical advantage.
% More formally,
% correlation clustering
% is defined as a graph clustering problem,
% where every edge of the graph is labeled with a '$+$' or a '$-$'.
% The objective is to construct clusters so that
% whenever the edge between two vertices is labeled with a '$+$' ('$-$'),
% the vertices should be placed in the same (different) cluster(s).
% In the application originally described in \cite{BansalBC04}
% these labels were generated by a binary classifier
% based on item features.
% Moreover, setting the number of clusters in advance is not required
% in correlation clustering.

Our main result is
{\em a novel variant of correlation clustering}
that uses relative distance comparisons only,
and is thus particularly well suited for human computation.
The relative distance comparisons we consider
are expressed in terms of {\em triplets}:
out of a set of three items,
one is designated as an ``outlier'',
meaning that it's distance from the two other items is the largest.

\textbf{We make the following contributions in this paper:}
\begin{enumerate}
\item We define a variant of the \textsc{Correlation-Clustering} problem
that takes a set of relative distance comparisons as input (Section~\ref{sect:rcc}).
\item We analyse the problem, show it to be NP-hard (Section~\ref{sect:nphard}),
discuss its connections to the standard \textsc{Correlation-Clustering} problem (Section~\ref{sect:rcc_vs_cc}),
and propose an $O(\log |U|)$ approximation algorithm (Section~\ref{sect:approxalg}),
where $U$ is the number of items being clustered.
This establishes that our problem is not harder than
\textsc{Correlation-Clustering} on general weighted graphs
from the point of view of approximation.
\item We also present a practical, simple, and very efficient local search algorithm
for solving our problem (Section~\ref{sect:localsearch}),
and carry out a number of experiments where we compare our
approach against two alternative methods from literature (Section~\ref{sect:exp}).
\end{enumerate}

\section{Problem definitions and analysis}
Let $U$ denote a set of items,
and let $d$ denote a distance function
$d\colon U \times U \rightarrow \mathbb{R}_0^+$
between items in $U$.
The {\em triplet} $(a,b,c)$, with $a, b, c \in U$,
captures the following {\em relative distance comparison} based on $d$:
\[
d(a,b) \leq \min\{ d(a,c), d(b,c) \}.
\]
That is,
of the items $a$, $b$, and $c$, item $c$ is an {\em outlier}.
(In our notation,
the outlier is always the third item of the triplet.)
Let $\triplets$ denote a set of such triplets over $U$.
Our task is to cluster the items in $U$ given only $\triplets$.
In informal terms,
we want to partition $U$ to disjoint subsets so that
items in the same subset are similar to each other,
and items in different subsets are different from one another
in terms of the distance function $d$.
Let $[1{:}n]$ denote the integers from $1$ to $n$.
A {\em clustering function} $f\colon U \rightarrow [1{:}k]$
assigns to every item in $U$ a {\em cluster label}, i.e.,
an integer between $1$ and $k$.
Let $\I\{ \cdot \}$ denote the indicator function.

\subsection{Background about correlation clustering}
%% Maybe should be put in the intro?
The problem we define in this paper
is closely related to the \textsc{Correlation-Clustering} problem,
as originally defined in \cite{BansalBC04}.
An instance of \textsc{Correlation-Clustering} is
a graph $G$ where
every edge $(u,v)$ is associated some positive weight $w(u,v)$,
as well as either the label $+$ or the label $-$.
The objective is to find a clustering of the vertices
so that vertices connected by a $+$ edge belong to the same cluster,
while vertices connected by a $-$ edge are assigned to different clusters.
More formally:
\begin{problem} (\textsc{Correlation-Clustering})
\label{prob:cc}
Given the graph $G = (U, \{E^+ \cup E^-\}, w)$,
where $U$ is a set of items,
$E^+$ and $E^-$ denote sets of edges
that are labeled with a $+$ and $-$, respectively,
and the edge weights are given by the function $w\colon E^+ \cup E^- \rightarrow \mathbb{R}^+$.
Find a clustering function $f$ that
minimizes the cost
\[
\begin{split}
c(f, G) = \sum_{(u,v) \in E^+} & w(u,v)\I\{ f(u) \neq f(v) \} \; + \\
& \sum_{(u,v) \in E^-} w(u,v)\I\{f(u) = f(v)\}.
\end{split}
\]
\end{problem}
The variant of \textsc{Correlation-Clustering}
given in Problem~\ref{prob:cc}
is commonly known as the one that
concerns {\em ``minimising disagreements in general weighted graphs''}.
Problem~\ref{prob:cc} is NP-hard \cite{BansalBC04},
as well as APX-hard \cite{DemaineEFI06}.

\subsection{Correlation clustering with relative distances}
\label{sect:rcc}
\begin{figure}
\centering
\includegraphics[width=0.6\columnwidth]{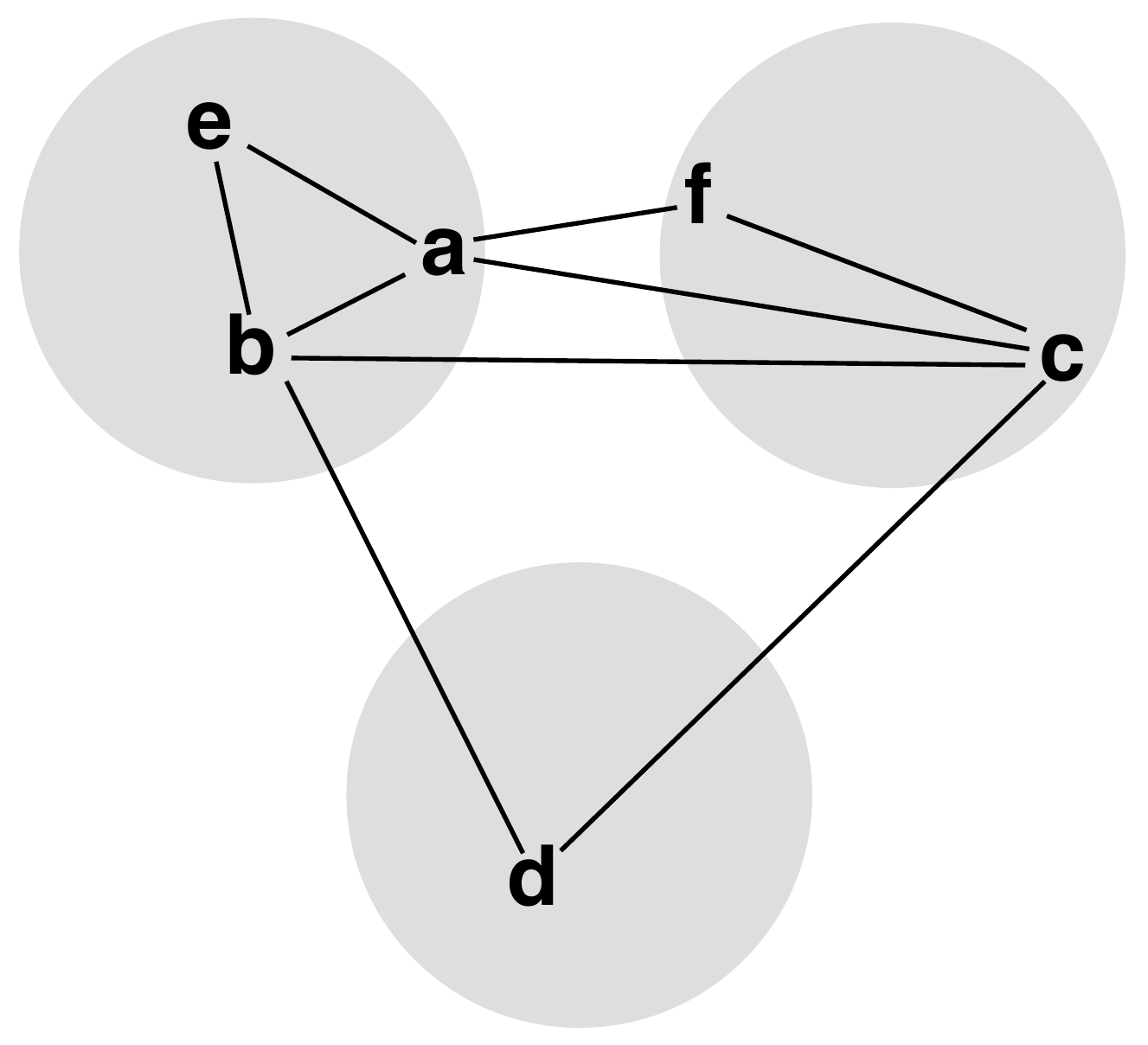}
\caption{A two-dimensional example with
a ground truth clustering of size 3
(the gray circles), items $a$, $b$, $c$, $d$, $e$, and $f$
at the locations shown,
as well as the
triplets $(a,b,c)$, $(b,d,c)$, $(a,b,e)$, and $(a,f,c)$.}
\label{fig:example_1}
\end{figure}
We proceed to
define a novel variant of \textsc{Correlation-Clustering}
that is based on relative distance comparisons
in the form of a set $\triplets$ of triplets as defined above.
Given $\triplets$, the task is to compute a clustering function $f$.
We first discuss
how relative distance comparisons
and an underlying ground truth clustering can interact.

Simply put (and somewhat exaggerated),
a triplet $(a,b,c)$ can be understood as saying that
{\bf ``$a$ and $b$ are next to each other, but $c$ is further away''}.
From the point of view of a clustering task,
we argue that
$(a,b,c)$ thus provides three (uncertain) pieces of information:
{\bf
\squishlist
\item[1.] items $a$ and $b$ might be in the same cluster,
\item[2.] items $a$ and $c$ might be in different clusters, and
\item[3.] items $b$ and $c$ might be in different clusters.
\squishend
}
In practice this will not always be true, of course.
Figure~\ref{fig:example_1} illustrates this
with a toy example
with three clusters, six items (letters $a$ to $f$), and four triplets.
All four triplets reflect the basic Euclidean distance
between the items.
Of the triplets,
$(a,b,c)$ ``correctly''
reflects the cluster structure according to the claim above:
$a$ and $b$ indeed are in the same cluster,
while $c$ is in a different cluster.
The remaining three triplets
show how conflicts arise between
the ground truth clustering and
relative distance comparisons.

First,
distances between all items in the triplet may be ``long'',
as is the case with $b$, $c$, and $d$,
and intuition says that the items should be put in different clusters.
In Figure~\ref{fig:example_1} this is also the case.
But if we interpret the triplet $(b,d,c)$ as above,
we would put $b$ and $d$ in the same cluster.
Second,
all pairwise distances may be ``short'',
as is the case with $a$, $b$ and $e$,
and again intuitively it would make sense to put $a$, $b$ and $e$ in the same cluster.
In Figure~\ref{fig:example_1} these items indeed do belong to the same cluster,
but our interpretation of $(a,b,e)$ suggests that $e$ is in a different cluster.
Third,
it is possible that in an optimal clustering,
items that are closer to each other in fact belong to different clusters,
while the outlier belongs to the same cluster as one of the other items.
This is the case with items $a$, $c$ and $f$.
Item $c$ is an outlier, but nonetheless belongs
to the same cluster as $f$ in
the ground truth solution of Figure~\ref{fig:example_1},
while the above interpretation of $(a,f,c)$ would
put $a$ and $f$ in the same cluster, and $c$ in a different cluster.

We leave a more fine-grained analysis of
the effects of these observations for future work,
and in this paper simply assume that
the three pieces of information provided by the triplet $(a,b,c)$
({\bf points 1--3} listed above)
can be used to identify useful clusterings.
However,
triplets inherently seem to make the clustering task more challenging
than the $E^+$ and $E^-$ constraints of Problem~\ref{prob:cc}.
Because, in the absence of noise,
$E^+$ and $E^-$ will always yield a clustering function $f$
such that $c(f, G) = 0$,
i.e.~the hardness of Problem~\ref{prob:cc} is caused by
noisy constraints.
(Indeed, without noise we can simply remove all edges in $E^-$
and are left with cliques of $E^+$ edges that correspond to the clusters.)
With triplets
no such easy solutions exist even in a noise-free scenario,
as the triplets do not
explicitly specify what items belong to the same cluster.

Moving on,
we say that the triplet $(a,b,c)$ is {\em satisfied}
by the clustering function $f$,
if and only if
$\I\{f(a) = f(b) \neq f(c)\}$
is true.
Otherwise $(a,b,c)$ is {\em unsatisfied}.
In Figure~\ref{fig:example_1}
all triplets except $(a,b,c)$ are unsatisfied by
the clustering shown.
We consider the following problem:
\begin{problem}
\label{prob:rcc}
%% (\textsc{Relative-Distance-Correlation-Clustering}):
Given a set of triplets $\triplets$ over items in the set $U$,
find a clustering function $f^\opt$ that minimizes the cost
\[
\begin{split}
s(f, \triplets) = \sum_{(a,b,c) \in \triplets} \I\{ \; & f(a) \neq f(b) \vee \\
& f(a) = f(c) \; \vee \; f(b) = f(c) \; \},
\end{split}
\]
i.e., $f^\opt = \arg\min_f s(f, \triplets)$.
\end{problem}
We have defined the objective function in Problem~\ref{prob:rcc}
to minimize the number of unsatisfied triplets.
% The problem of maximising the number of satisfied triplets
% can be defined in the same way.
Observe that the number of clusters is not specified
as part of the problem input.
As with traditional \textsc{Correlation-Clustering},
$f^\opt$ may use any number of clusters that minimizes $s(f^\opt,\triplets)$.

\subsection{Problem complexity}
\label{sect:nphard}
%% We first discuss the complexity of Problem~\ref{prob:rcc}.
\begin{theorem}
Problem~\ref{prob:rcc} is NP-hard.
\end{theorem}
\begin{proof}
The proof is a simple reduction from the unweighted variant of Problem~\ref{prob:cc}.
(This problem is also NP-hard.)
Given an instance of \textsc{Correlation-Clustering}
determined by the edge sets $E^+$ and $E^-$,
we create an instance of Problem~\ref{prob:rcc}, i.e.,
a set of triplets $\triplets$ as follows:
for every $(u,v) \in E^+$, insert the triplet $(u,v,x^{uv})$,
and for every $(u,v) \in E^-$ insert the triplet $(u, u_v^\prime, v)$.
Here $x^{uv}$ and $u^\prime$ are dummy items that each occur
in a single triplet only, and can hence be satisfied trivially.
The only real constraints to $f^\opt$ are thus
determined by items that appear in $E+$ and $E-$.
Observe that an optimal solution $f^\opt$ to Problem~\ref{prob:rcc}
immediately gives an optimal solution to
the \textsc{Correlation-Clustering} instance.
Moreover, the costs of the solutions are identical.
\end{proof}
The reduction employed in the proof above
has the implication that
Problem~\ref{prob:rcc} is at least as hard as Problem~\ref{prob:cc}
(with unit weights)
also from the point of view of approximation.
Indeed,
any approximation bound for Problem~\ref{prob:rcc}
also holds for the unweighted variant of Problem~\ref{prob:cc}.
%% THIS IS STILL TO BE CHECKED!
% Since \textsc{Correlation-Clustering} has been shown to be APX-hard \cite{xxx},
% we can conclude that Problem~\ref{prob:rcc} must be at least this hard as well.

%%\subsection{Mapping Problem~\ref{prob:rcc} to basic \textsc{Correlation-Clustering}}
\subsection{Mapping Problem~\ref{prob:rcc} to Problem~\ref{prob:cc}}
\label{sect:rcc_vs_cc}
We continue by discussing further properties of Problem~\ref{prob:rcc}.
These will be useful
when we design algorithms for the problem.
In short,
we show how to ``clean up'' the set of triplets $\triplets$
so that it is possible to construct
a reguler \textsc{Correlation-Clustering} instance
from the triplets.

First,
observe that if $\triplets$ contains both
triplets $t = (u,v,x_1)$ and $t^\prime = (u,x_2,v)$ for any $u, v, x_1, x_2 \in U$,
only one of these can be satisfied by {\em any} clustering function $f$
(including $f^\opt$),
because to satisfy $t$ we must have $f(u) = f(v)$,
while satisfying $t^\prime$ requires $f(u) \neq f(v)$.
Some of the triplets in $\triplets$ can (and in most cases will) thus be inconsistent
with each other.
This is a natural consequence of 
how relative distance information is encoded in the triplets:
items $u$ and $v$ may be ``close'' to each other in relation to item $x_1$,
but ``far apart'' when compared with item $x_2$.
Indeed, inconsistencies occur also in sets of triplets $\triplets$
that are fully noiseless,
as also discussed above.
These inconsistencies are naturally represented
in terms of a {\em constraint graph}.
\begin{definition}
\label{def:cgraph}
Two triplets are inconsistent when for some items $u$ and $v$
and any clustering function $f$,
one of the triplets is satisfied when $f(u) = f(v)$,
while the other is satisfied when $f(u) \neq f(v)$.
Let $C_\triplets = (\triplets, E)$ denote the constraint graph
associated with the set $\triplets$.
The vertices of $C_\triplets$ are the triplets,
and the edge set $E$ contains those pairs of triplets
that are inconsistent with each other.
\end{definition}

A {\em vertex cover} of a graph is
a subset of its vertices
such that for every edge,
at least one endpoint belongs to the vertex cover.
Consider a vertex cover of $C_\triplets$, denoted $\vc(C_\triplets)$.
(Notice that {\em any} vertex cover will do, it does not have to be one of minimum size.)
Let $\triplets^\prime = \triplets \setminus \vc(C_\triplets)$,
i.e., $\triplets^\prime$ contains those triplets that are not part of
the vertex cover $\vc(C_\triplets)$.
We can show that
%% in contrast to $\triplets$,
in $\triplets^\prime$ there are no
triplets that would provide conflicting information about any two
items $u$ and $v$:
\begin{lemma}
\label{lemma:inconsistency}
There are no inconsistent triplets in $\triplets^\prime = \triplets \setminus \vc(C_\triplets)$.
\end{lemma}
\begin{proof}
See Appendix~\ref{app:inconsistency}.
\end{proof}
In practice this implies that
all triplets in $\triplets^\prime$ are guaranteed to
consistently suggest that items $u$ and $v$ either belong to the same cluster,
or that items $u$ and $v$ belong to different clusters.
Importantly, there are {\em no} two triplets in $\triplets^\prime$
such that one would prefer $u$ and $v$ in the same cluster,
while the other would prefer $u$ and $v$ in different clusters.
We can thus map $\triplets^\prime$ to an instance of
\textsc{Correlation-Clustering} (Problem~\ref{prob:cc}).
\begin{definition}
\label{def:residualcc}
Given the set of triplets $\triplets^\prime = \triplets \setminus \vc(C_\triplets)$,
we define the instance
$G_{\triplets^\prime} = (U, \{E^+ \cup E^-\}, w)$
of \textsc{Correlation-Clustering}
as follows:
\squishlist
\item The set of items $U$ is the union of all items in the triplets in $\triplets^\prime$.
\item For every pair $u, v \in U$,
let $\triplets^\prime_{u,v}$ denote those triplets in $\triplets^\prime$
that contain both items $u$ and $v$.
% As argued above,
% all triplets in $\triplets^\prime_{u,v}$ must agree that
% neither $u$ or $v$ is the outlier,
% or that one and only one of $u$ or $v$ is the outlier.
\item If for given $u$ and $v$, all triplets in $\triplets^\prime_{u,v}$ agree that
neither $u$ or $v$ is the outlier,
we include $\{u,v\}$ into $E^+$.
\item If for given $u$ and $v$, all triplets in $\triplets^\prime_{u,v}$ agree that
either $u$ or $v$ must be the outlier,
we include $\{u,v\}$ into $E^-$.
\item The weight $w(u,v)$ is the size of $\triplets^\prime_{u,v}$.
\squishend
\end{definition}
Because there are no inconsistent triplets in $\triplets^\prime$,
every pair $\{u,v\}$ will be assigned to either $E^+$ or $E^-$.

\section{An approximation algorithm}
\label{sect:approxalg}
Above we showed that
by finding a vertex cover of the constraint graph $C_\triplets$,
we can turn a given instance of Problem~\ref{prob:rcc} into
a regular \textsc{Correlation-Clustering} instance.
We make use of this to design
an approximation algorithm for Problem~\ref{prob:rcc}.

First, we assume that an approximation algorithm exists for \textsc{Correlation-Clustering}
in general weighted graphs.
\begin{assumption}
\label{asm:ccapp}
Let $f^\opt_{CC}$ denote the optimal solution to
the \textsc{Correlation-Clustering} instance $G_{\triplets^\prime}$.
We assume there exists a polynomial time algorithm
for \textsc{Correlation-Clustering} that finds the solution $f^{\alg}$
that satisfies
\[
c( f^{\alg}, G_{\triplets^\prime} ) \leq \alpha \; c( f^\opt_{CC}, G_{\triplets^\prime} ),
\]
where $\alpha$ is some function of the input $\triplets$.
\end{assumption}
Second, we assume that a constant factor approximation algorithm exists for
finding minimum vertex covers:
\begin{assumption}
\label{asm:vcapp}
Let $C_\triplets$ denote the constraint graph of Definition~\ref{def:cgraph},
and let $\vc_{\min}(C_\triplets)$ denote its vertex cover of {\em minimum} size.
We assume there exists a polynomial time algorithm
that finds the solution $\vc^{\alg}(C_\triplets)$
that satisfies
\[
|\vc^{\alg}(C_\triplets)| \leq \beta \; |\vc_{\min}(C_\triplets)|,
\]
where $\beta > 1$ is some constant.
\end{assumption}
Our approximation algorithm is described in Algorithm~\ref{alg:1}.
In short, we first solve minimum vertex cover on $C_\triplets$,
discard triplets that belong to the found cover,
and then solve the remaining \textsc{Correlation-Clustering} instance.
The algorithm runs in polynomial time
as long as the algorithms of assumptions \ref{asm:ccapp} and \ref{asm:vcapp}
run in polynomial time.
The intermediary steps are trivially polynomial in the size of $\triplets$.
\begin{algorithm}[t]
\caption{An approximation algorithm}
\label{alg:1}
\begin{algorithmic}[1]
\STATE Input: triplets $\triplets$.
\STATE Construct the constraint graph $C_\triplets$ according to Definition~\ref{def:cgraph}.
\STATE Find an approximate minimum vertex cover $\vc^{\alg}(C_\triplets)$ by using the algorithm of Assumption~\ref{asm:vcapp}.
\STATE Let $\triplets^\prime \leftarrow \triplets \setminus \vc^{\alg}(C_\triplets)$.
\STATE Construct $G_{\triplets^\prime}$ according to Definition~\ref{def:residualcc}.
\STATE Find an approximate solution $f^{\alg}$ to $G_{\triplets^\prime}$ by using the algorithm of Assumption~\ref{asm:ccapp}.
\STATE Return $f^{\alg}$.
\end{algorithmic}
\end{algorithm}

We give the following theorem:
\begin{theorem}
\label{thm:approx}
Let $f^\opt$ denote the optimal solution to Problem~\ref{prob:rcc},
and denote by $f^{\alg}$
the solution found by Algorithm~\ref{alg:1}.
We have
\[
s( f^{\alg}, \triplets ) \leq (2\alpha + \beta) \; s( f^\opt, \triplets ),
\]
where $\alpha$ and $\beta$ are
the approximation factors in assumptions
\ref{asm:ccapp} and \ref{asm:vcapp}, respectively.
\end{theorem}
\begin{proof}
See Appendix~\ref{app:approxproof}.
\end{proof}
The actual bound thus depends on both $\alpha$ and $\beta$.
Currently best known approximation algorithms
for solving \textsc{Correlation-Clustering} in general weighted graphs
and finding minimum vertex covers
have bounds of $\alpha = O(\log n)$ \cite{DemaineEFI06} and
$\beta = 2 - \frac{\log\log n}{2\log n}$ \cite{MonienS85}, respectively.
%%For simplicity, we let $\beta = 2$ in the following.
The graph associated with the \textsc{Correlation-Clustering} instance $G_{\triplets^\prime}$
has $|U|$ vertices, and hence we have $\alpha = O(\log |U|)$ and obtain:
\begin{corollary}
Algorithm~\ref{alg:1} is an $O(\log |U|)$ approximation algorithm
for Problem~\ref{prob:rcc}.
\end{corollary}
From the point of view of approximation,
Problem~\ref{prob:rcc} is thus not asymptotically harder than Problem~\ref{prob:cc}
(after omitting constants).
Indeed, Algorithm~\ref{alg:1}
is mainly of theoretical interest,
as it shows that Problem~\ref{prob:rcc},
like regular \textsc{Correlation-Clustering},
admits an approximation bound.

\section{A local search algorithm}
\label{sect:localsearch}
Next, we describe a more practical method,
at the core of which is a local search heuristic.
It is,
however,
to a certain extent inspired by Algorithm~\ref{alg:1},
as will become apparent below.

\subsection{Outline of the algorithm}
In short,
we minimize the cost function $s(f, \triplets)$
using a {\bf simple greedy local search algorithm}.
The algorithm updates the cluster assignment of
a single item $u \in U$ at a time,
while keeping the cluster assignment of all other items fixed.
The algorithm makes passes over all items
%% and updates $f$ one item at a time
until it reaches a fixed point
where the value of $f(u)$ no longer changes for any $u \in U$.
Details are shown in Algorithm~\ref{alg:localsearch}.

\begin{algorithm}[t]
\caption{A local search heuristic}
\label{alg:localsearch}
\begin{algorithmic}[1]
\STATE \textbf{Input:} a set of triplets $\triplets$
\STATE $f \leftarrow $ initialise (can be done in different ways)
\STATE $\knew \leftarrow \max_{u \in U} f(u) + 1$
\STATE $f^\prime \leftarrow \emptyset$
\WHILE{$f^\prime \neq f$}
  \STATE $f^\prime \leftarrow f$
  \FOR{$u \in U$}
    \STATE $f(u) \leftarrow \arg\min_{f(u) \in [1{:}\knew]} s(f,\triplets)$
    \IF{$f(u) = \knew$}
      \STATE $\knew \leftarrow \knew + 1$
    \ENDIF
  \ENDFOR
  \STATE ``clean up'' $f$ so that it maps $U$ to the range $[1{:}h]$, where $h = |\{f(u)\}_{u \in U}|$.
  \STATE $\knew \leftarrow \max_{u \in U} f(u) + 1$
\ENDWHILE
\STATE \textbf{return} $f$
\end{algorithmic}
\end{algorithm}

We can initialise $f$ in different ways on line 2 of Alg.~\ref{alg:localsearch}.
In this paper we consider two approaches:
\begin{enumerate}
\item \textbf{All equal:} We set $f(u) = 1$ for every $u \in U$.
\item \textbf{All different:} We initialise $f$ to be a bijection from $U$ to the integers $[1{:}|U|]$, that is, every $i \in [1{:}|U|]$ is the initial cluster assignment one and only one $u \in U$.
\end{enumerate}
When updating $f(u)$ on line 8,
the algorithm considers all possible values in the range $[1{:}\knew]$.
Here $\knew$ is
the index of a new cluster
that does not yet exist in $f$.
The update step can thus introduce new clusters to $f$.
On line 13 the algorithm
re-assigns (``cleans up'') $f$ so that there are no gaps in the cluster indices,
meaning that the cluster indices must range from $1$ to the number of clusters.

\subsection{Practical considerations}
As discussed above in Section~\ref{sect:rcc_vs_cc},
the set $\triplets$ of triplets
may (and in practice will) contain inconsistencies.
These inconsistencies are 
a fundamental property of relative distance comparisons.
Above we also showed that
by finding a vertex cover of the constraint graph $C_\triplets$,
we can remove such inconsistencies.
While the local search heuristic described here
should in theory be unaffected by
these inconsistencies,
it is conceivable that
we will in practice obtain better results
after making $\triplets$ consistent.
That is,
in the experiments we will
consider an approach where the input to Algorithm~\ref{alg:localsearch}
is in fact $\triplets^\prime = \triplets \setminus \vc(C_\triplets)$.

This requires computing a vertex cover of $C_\triplets$.
We will employ the simple and well-known 2-approximation algorithm\footnote{This algorithm was independently proposed by F.~Gavril and M.~Yannakakis, according to \cite{PapadimitriouS82}.}
that selects edges from the input graph one by one,
while at every step removing all other edges adjacent to the selected edge (including the selected edge)
and adding the endpoints of the selected edge to the cover.
The algorithm, called \vcapprox, terminates when the set of edges becomes empty.
(For details, see e.g.~page 1025 of \cite{CormenLRS01}.)

Despite the constant factor approximation guarantee,
this algorithm can produce a fairly large cover.
In theory this does not really matter,
since in our use-case
{\em any} cover of $C_\triplets$ will result in a set of triplets $\triplets^\prime$
that is free of inconsistencies.
However,
in practice we would like there to be
a sufficient amount of information about the relative distances,
and hence $\triplets^\prime$ should remain as large as possible.
(Indeed, an empty set of triplets is free of inconsistencies,
but it is also rather useless.)

We will thus use a
simple heuristic
to further reduce the size of the cover produced by \vcapprox.
The idea is to remove all such vertices from the cover
that are redundant.
{\em A vertex is redundant in a vertex cover
if all of its neighbour vertices also belong to the cover.}
We thus consider every vertex in $\vc(C_\triplets)$ one-by-one,
and remove it from the cover if (and only if) all of its neighbours
belong to the cover.
This leaves us with a proper vertex cover of $C_\triplets$,
but the size of the cover will in practice be substantially reduced.

\section{Experiments}
\label{sect:exp}
We conduct experiments with the following variants of our local search heuristic:
\begin{itemize}
\item \textsf{Ls-EQ}: Algorithm~\ref{alg:localsearch} with the \textbf{All equal} initialisation of $f$.
\item \textsf{Ls-AD}: Algorithm~\ref{alg:localsearch} with the \textbf{All different} initialisation of $f$.
\item \textsf{Ls-EQ-VC}: Algorithm~\ref{alg:localsearch} with the \textbf{All equal} initialisation of $f$. The input triplets are first cleaned up by running the vertex cover heuristic on $C_\triplets$.
\item \textsf{Ls-AD-VC}: Algorithm~\ref{alg:localsearch} with the \textbf{All different} initialisation of $f$.  The input triplets are first cleaned up by running the vertex cover heuristic on $C_\triplets$.
\end{itemize}
Of these \textsf{Ls-AD-VC} is the one that we mainly promote and study,
others are shown for comparison.
We implemented both the local search as well as the vertex cover heuristic in JavaScript.
All experiments were run with Node.js.

We also consider two alternative approaches.
Both first compute an embedding of the items,
and then run a ``standard'' clustering algorithm with this as the input.
The first method, called \textsf{CrowdClust} below,
is the method described in \cite{GomesWKP11}.
The second, called \textsf{t-STE},
is a well-known stochastic neighbourhood embedding method
adapted to relative distance judgements in \cite{MaatenW12}.
The clustering algorithm
is the same in both methods:
a Dirichlet Process mixture model (MB-VDP) \cite{GomesWP08}.
We chose this, because like our algorithms,
it does not require setting the number of clusters in advance.
% The comparison is thus concerning
% the two approaches to compute embeddings from triplets.
Of \textsf{CrowdClust} and MB-VDP
we use the implementations provided by Ryan Gomes\footnote{\texttt{http://www.vision.caltech.edu/~gomes/software.html}},
while of \textsf{t-STE} we use
Michael Wilber's implementation\footnote{\texttt{https://github.com/gcr/tste-theano}}.

Finally, we emphasise that
the experiments have been carried out to
{\em illustrate the experience of a naive user},
without any extensive parameter tuning for the \textsf{CrowdClust} nor \textsf{t-SNE} methods.
The number of optimisation iterations in \textsf{CrowdClust} was capped at 20,
and both methods compute a 4-dimensional embedding.
Otherwise we use default parameters suggested by the authors.
This is because we want to highlight the simplicity of
our algorithms that require
no parameter tuning of any kind,
and thus competing approaches should also work pretty much ``out of the box''.

\subsection{Experiment 1: Artificial data with known ground truth}
\label{sect:exp1}
\begin{figure*}[t]
\centering
\includegraphics[width=0.32\textwidth]{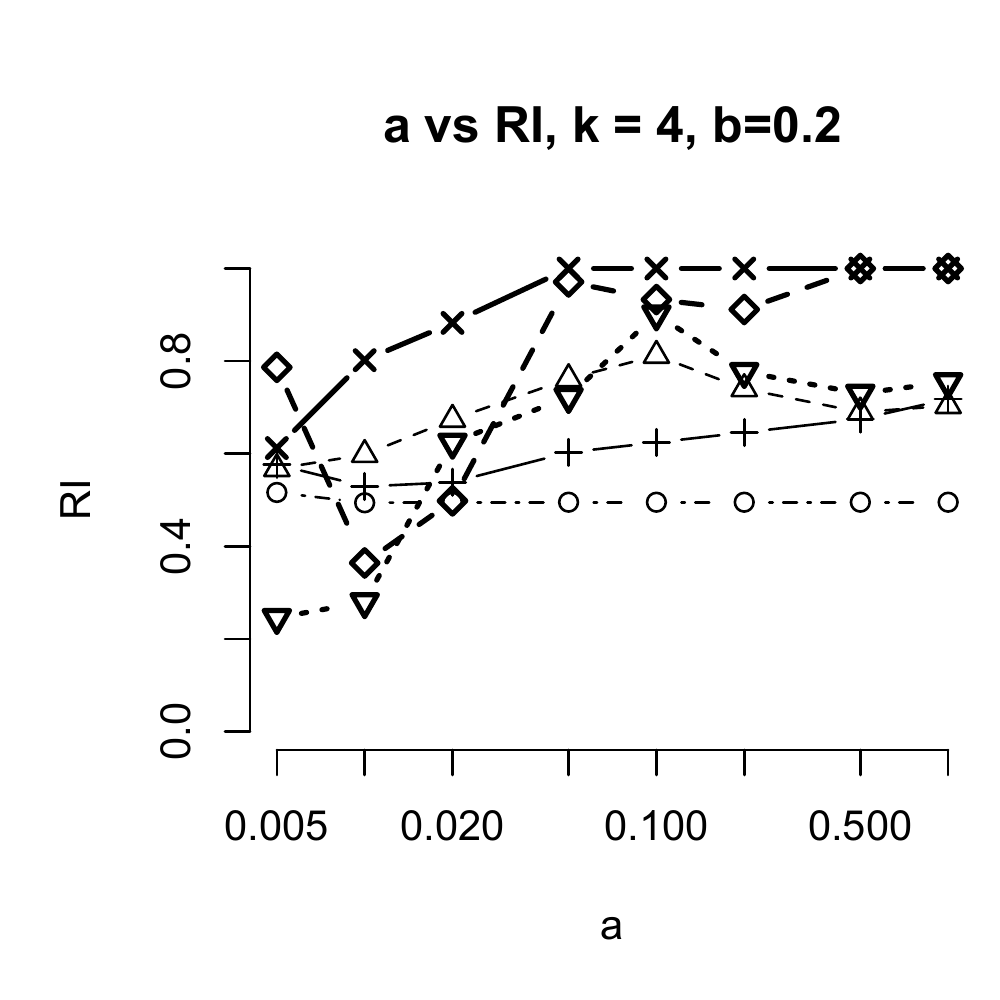}
\includegraphics[width=0.32\textwidth]{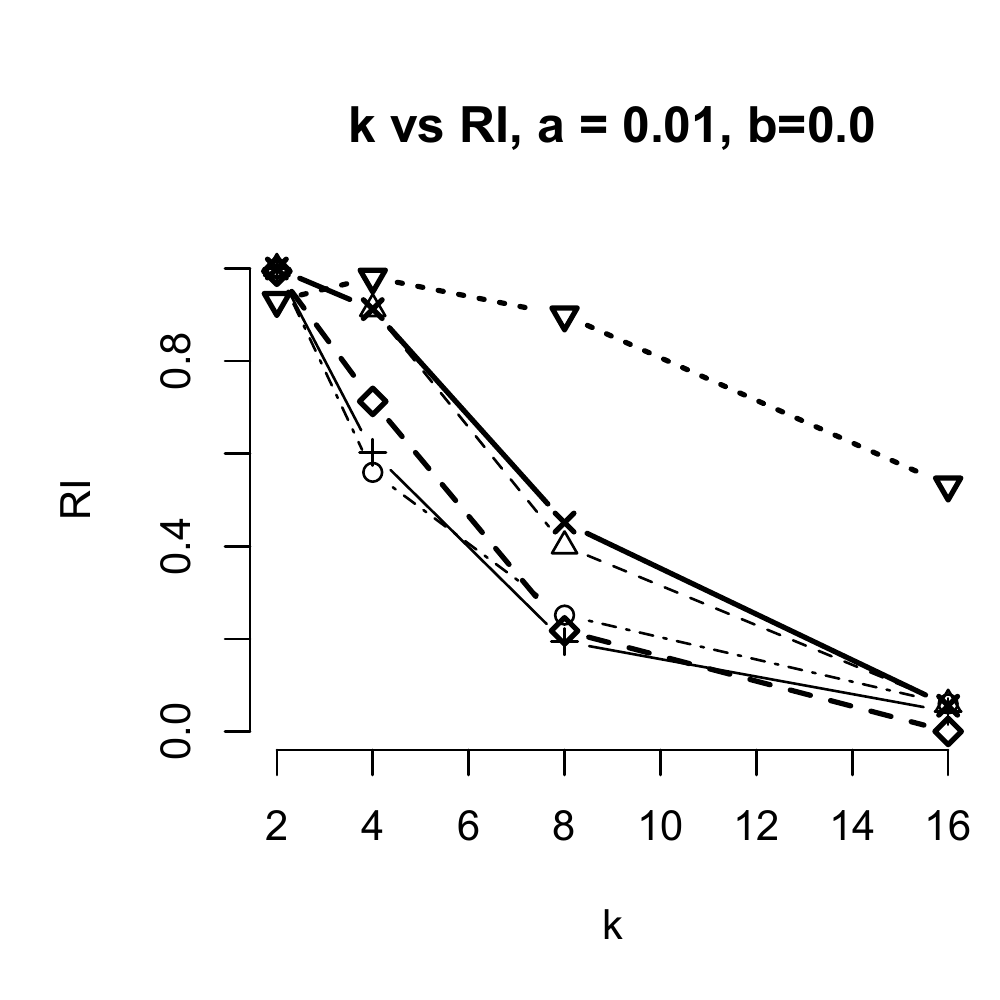}
\includegraphics[width=0.32\textwidth]{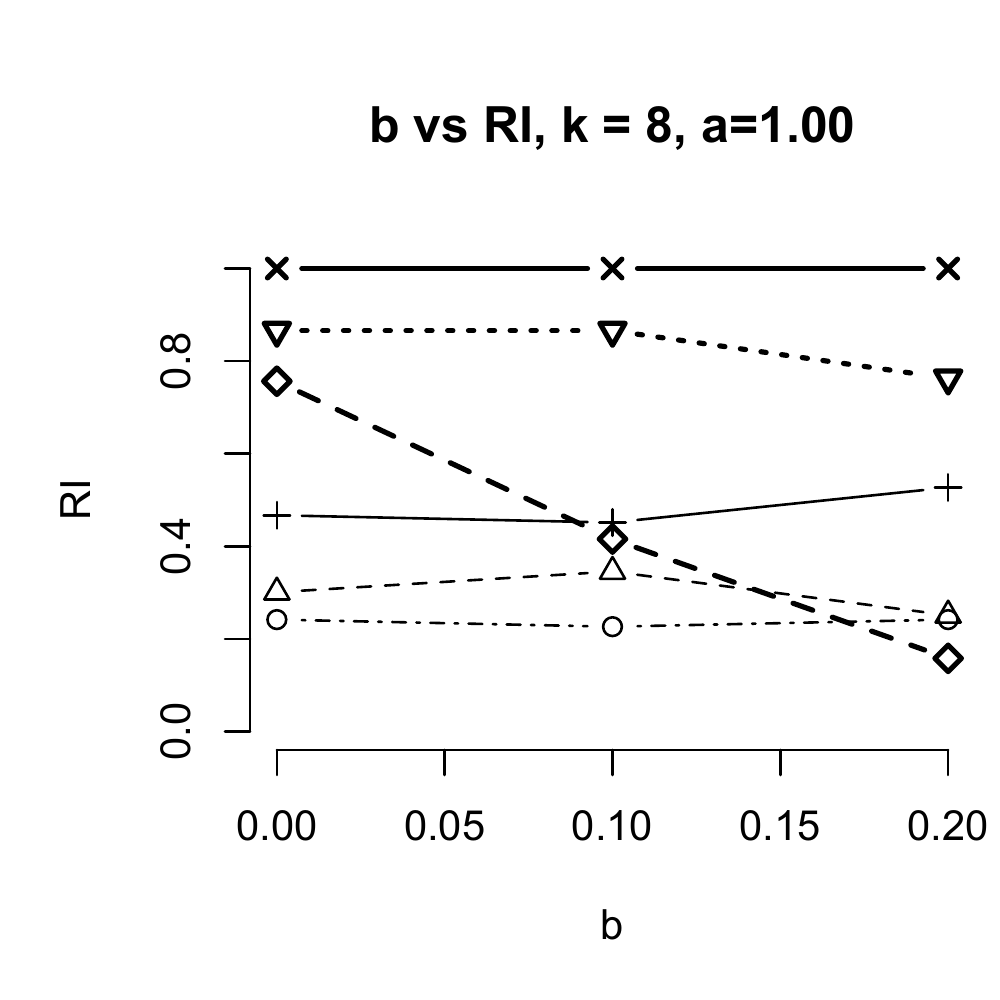}\\
\includegraphics[width=0.32\textwidth]{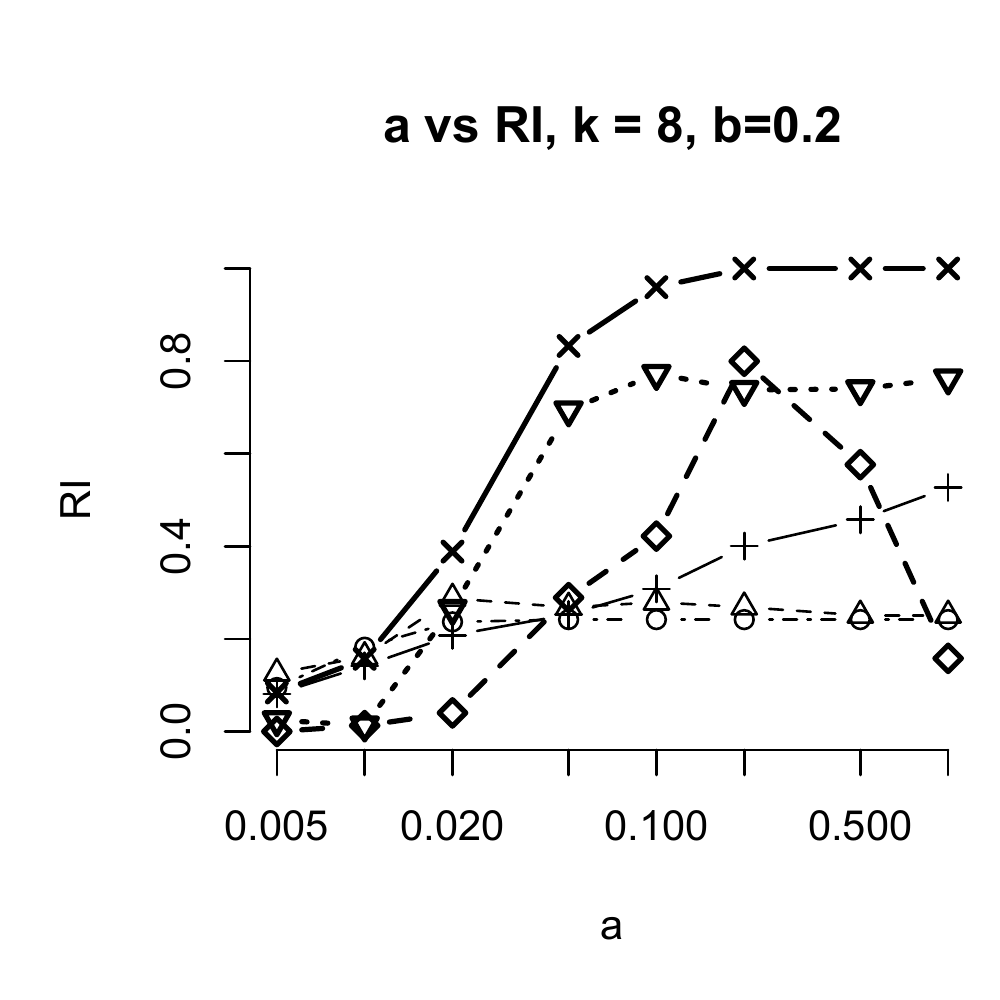}
\includegraphics[width=0.32\textwidth]{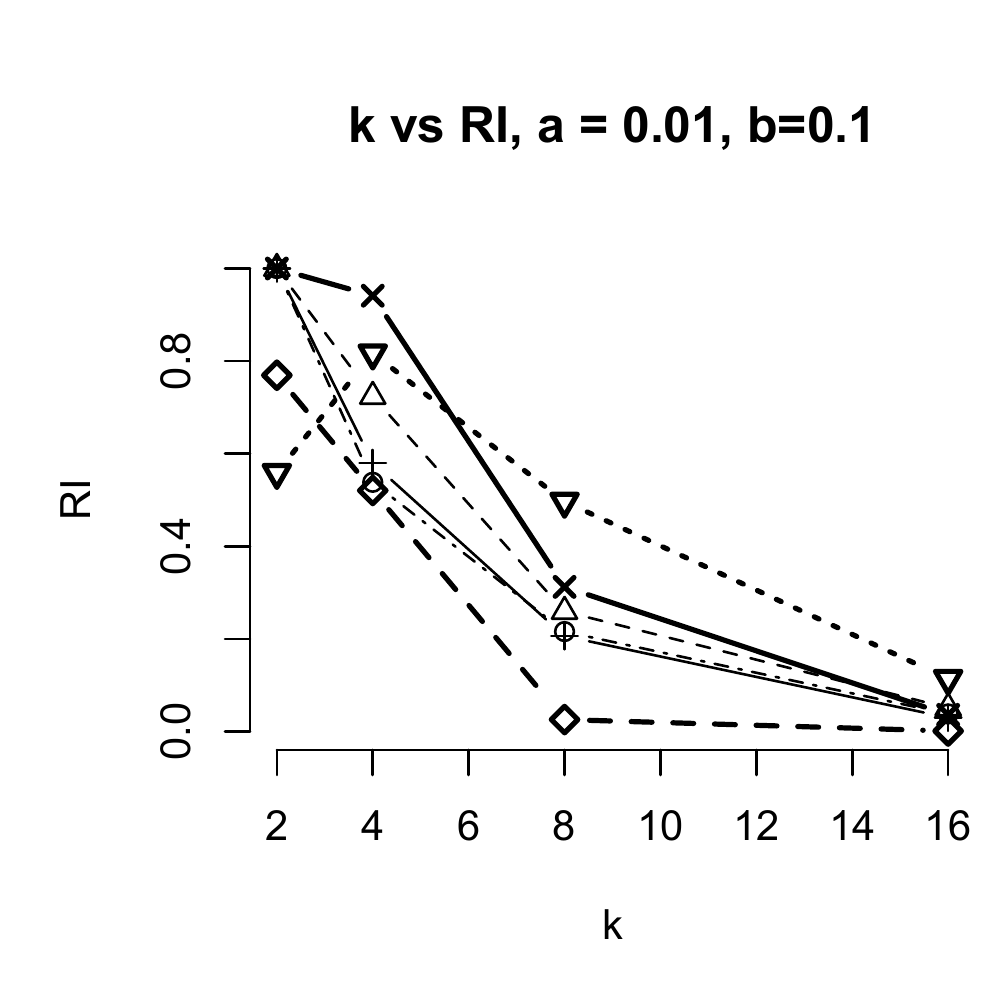}
\includegraphics[width=0.32\textwidth]{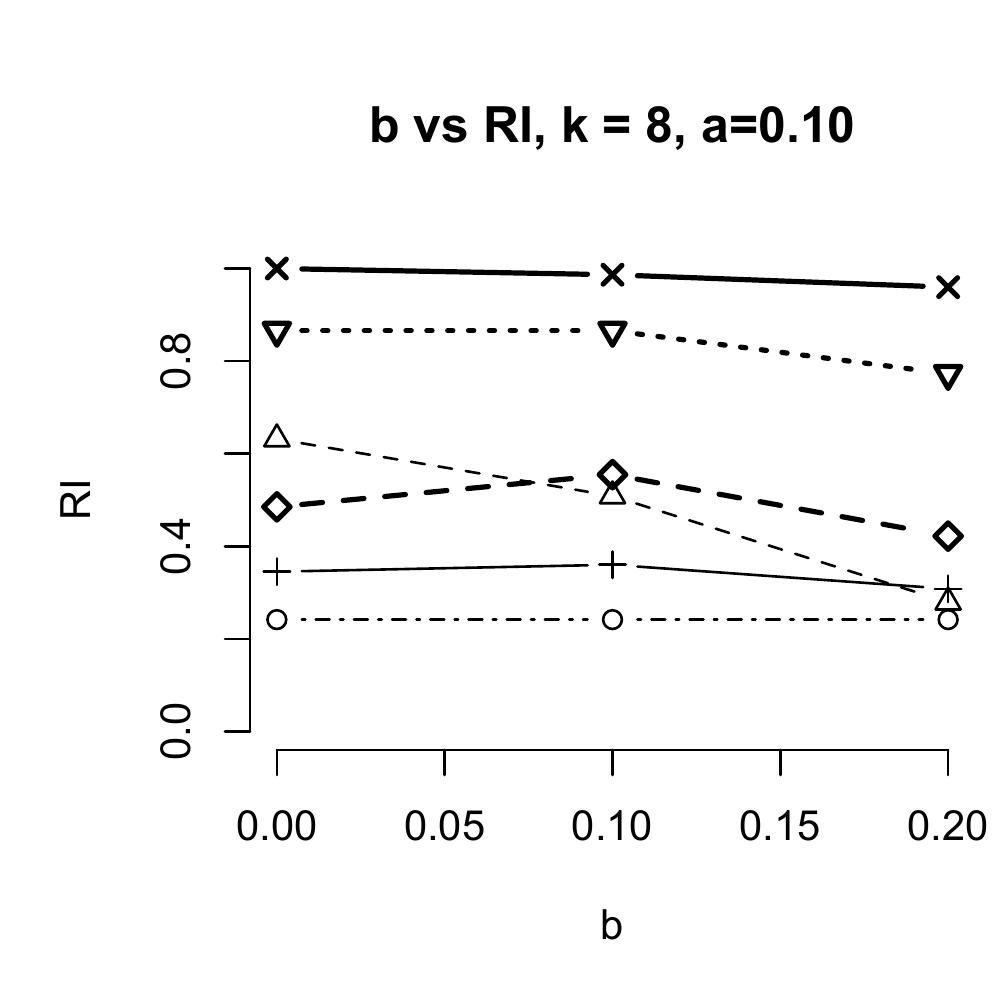}\\
\includegraphics[width=0.32\textwidth]{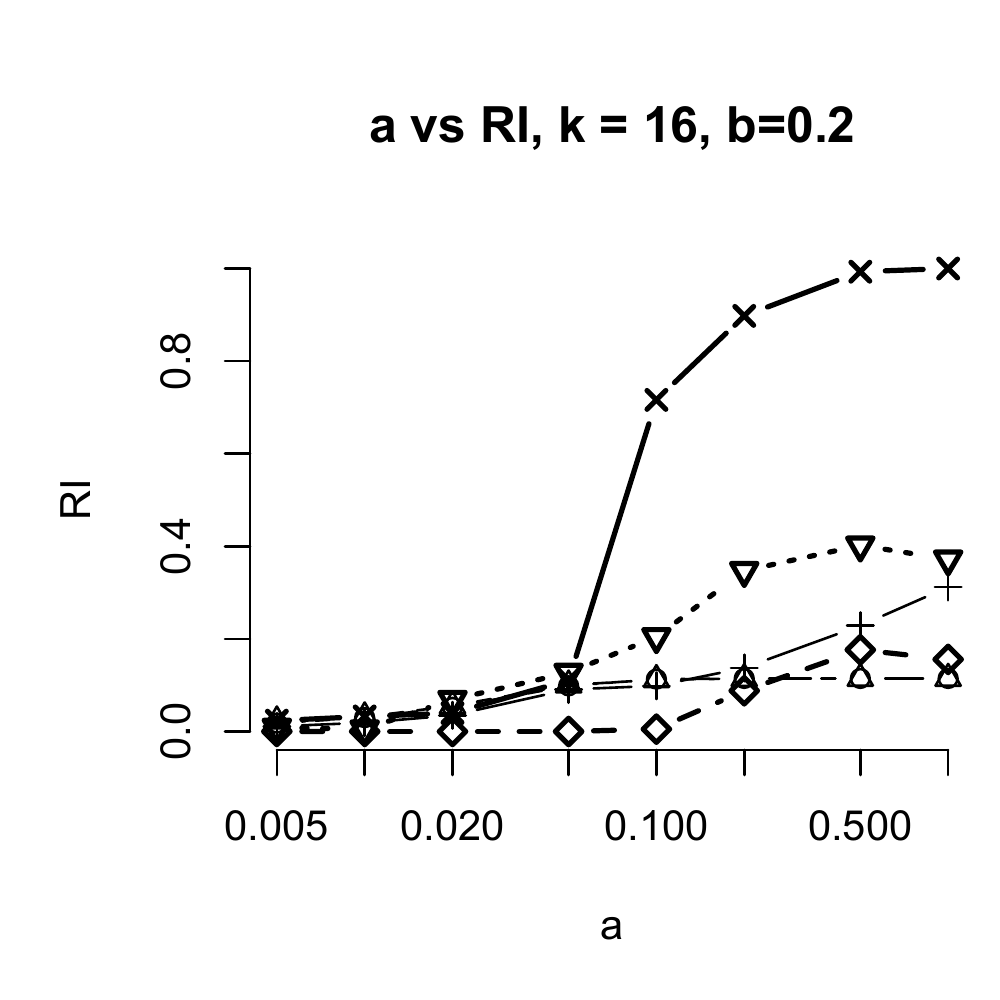}
\includegraphics[width=0.32\textwidth]{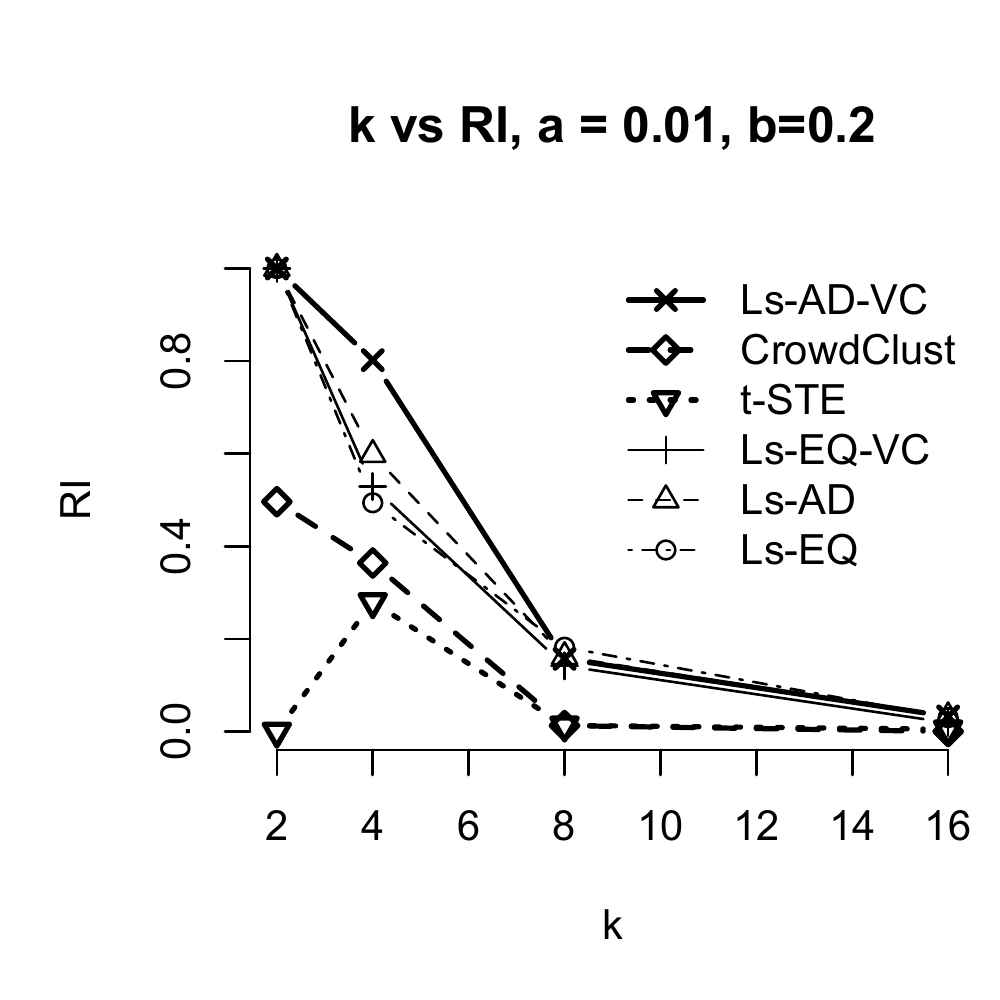}
\includegraphics[width=0.32\textwidth]{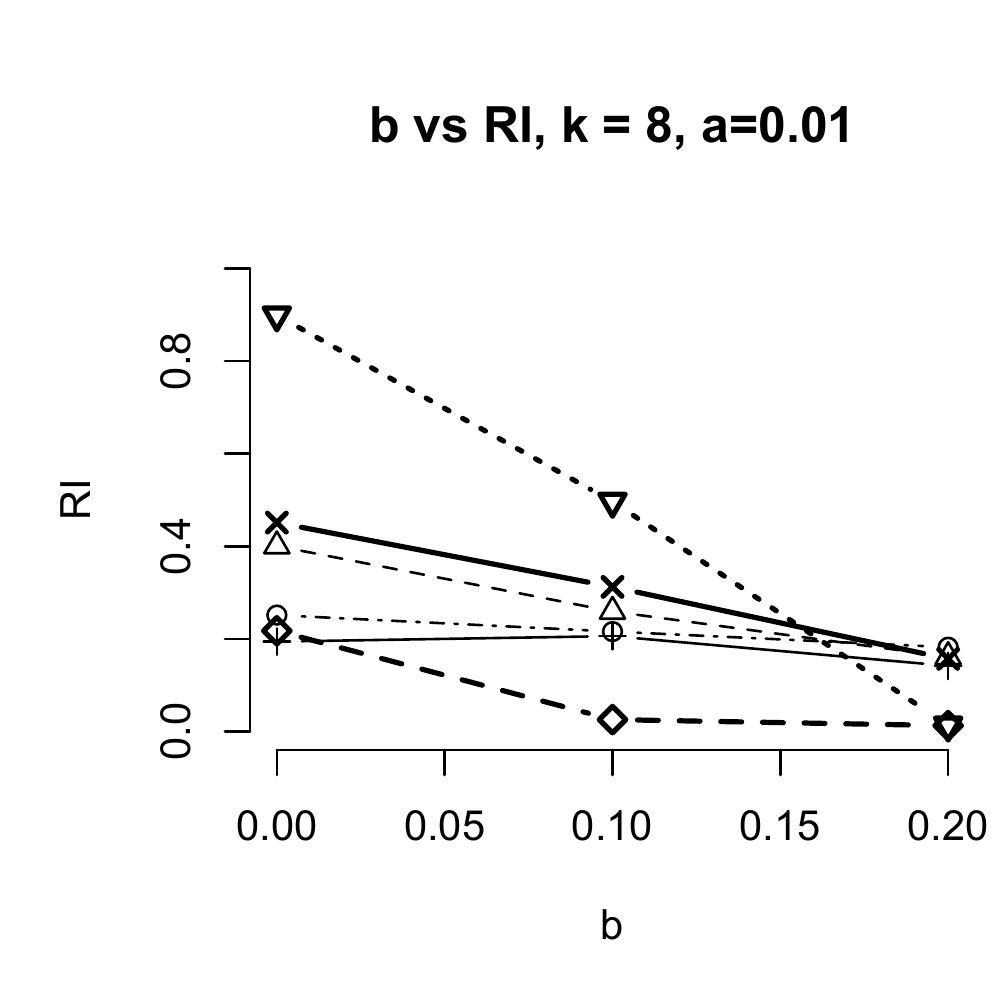}\\
\caption{Results of Experiment 1. The panels show the adjusted Rand index as a function of different triplet generator parameters ($a$ left, $k$ middle and $a$ right). See Section~\ref{sect:exp1} for details. Legend for all panels is shown in the center bottom panel.}
%% In short, we find that the correlation clustering based \textsf{Ls-AD-VC} algorithm has better chances than \textsf{CrowdClust} to reconstruct the underlying ground truth clustering when there are fewer triplets and the input is noisy.
\label{fig:exp1main}
\end{figure*}
\textbf{Setup:}
We generate a set $\triplets$ of triplets from
a known ground truth clustering $f^*$ over $160$ items,
and measure how well the algorithms can recover $f^*$ given the triplets.
The triplets are generated by first
constructing all possible triplets
(of which there are ${160 \choose 3} = 669,920$ in this case),
then selecting a random subset of these to include in $\triplets$,
and finally by introducing noise to some randomly chosen triplets
by swapping the outlying item with one of the two other items (chosen at random).
The process is thus parametrised by
three quantities:
\begin{itemize}
\item[$k$]: the number of clusters in $f^*$,
\item[$a$]: the fraction of triplets (out of all possible triplets) to include in $\triplets$, and
\item[$b$]: the fraction of ``noisy'' triplets in $\triplets$.
\end{itemize}
We let $k \in \{2, 4, 8, 16\}$,
$a \in \{0.005, 0.01, 0.02,0.05, 0.1,$ $0.2, 0.5, 1.0\}$,
and $b \in \{0, 0.1, 0.2\}$.
When $a = 1$ and $b = 0$
the process outputs all possible triplets without any noise.
With other choices of $a$ and $b$
the triplet generator is non-deterministic.
For each of such combination of $a$ and $b$
we generate 10 independent sets of triplets.
This results in 924 test cases (including repeated trials) in total.
We run the algorithms for all of these,
and report averages over the 10 trials
for given values of $a$, $b$ and $k$.
We evaluate the algorithms by comparing the found clusterings
to a ground truth in terms of the adjusted Rand index (RI, larger values better) \cite{hubert1985comparing}.

\textbf{Results:} We find that in 894 out of the 924
test cases ($\approx 97$\%), \textsf{Ls-AD-VC}
has at least the same RI value as \textsf{CrowdClust}.
When compared against \textsf{t-STE},
\textsf{Ls-AD-VC} performs at least as well in
751 out of 924 cases ($\approx 81$\%) in terms of RI.
That is,
in most cases our approach seems to
do a better job at
reconstructing the true clustering $f^*$.
When comparing \textsf{Ls-AD-VC} against \textsf{Ls-AD}
(the pure local search heuristic without
vertex cover based pre-processing),
we find that in 813 out of the 924 test cases ($\approx 88$ percent)
\textsf{Ls-AD-VC} has better (or the same) performance in terms of RI.
This suggests that
pre-processing $\triplets$ so that there are no
inconsistencies is useful.
Finally,
when comparing the two initialisation strategies
(all equal, all different),
we find that in 871 out of the 924 test cases ($\approx 94$ percent)
\textsf{Ls-AD-VC} outperforms \textsf{Ls-EQ-VC}.
The local search heuristic is thus more successful in reconstructing $f^*$
when it starts from a configuration where all points
are in different clusters.

A more fine-grained analysis
of how different parameters of
the triplet generator affect performance
is shown in Figure~\ref{fig:exp1main}.
The left, middle, and right columns
in Figure~\ref{fig:exp1main}
show RI
as a function of $a$, $k$ and $b$, respectively.

From the left column we observe that
when only a small fraction of triplets are available,
and there is a fair amount of erroneous triplets included in the input ($b = 0.2$),
all algorithms have difficulties
reconstructing the original $f^*$.
This is especially true as
the number of clusters in $f^*$ increases, as one might expect.
With $k=4$ (topmost panel)
both \textsf{Ls-AD-VC} and \textsf{CrowdClust}
can recover $f^*$ (almost) perfectly
when $a \geq 0.05$,
while \textsf{t-STE} does not do bad either.
When $k \geq 8$ (middle and bottom panel in left column) both
\textsf{Ls-AD-VC} and \textsf{t-STE} seem to
produce reasonable results (for $a \geq 0.1$),
with \textsf{Ls-AD-VC} finding
a near-perfect clustering as long as $a$ is large enough.

However, in practical human computation applications
small values of $a$ may be more relevant,
as this corresponds to fewer triplets, and hence less work by the human annotators.
This situation is highlighted in the middle column,
where we consider RI as a function of $k$
with fixed $a = 0.01$, and different values of $b$.
In the absence of noise ($b = 0$, top panel),
\textsf{t-STE} is a clear winner.
However, as noise is introduced,
\textsf{Ls-AD-VC},
as well as the other correlation clustering based methods,
outperform the two embedding approaches.
For a large number of clusters (say, $k \geq 8$),
the problem is very hard for all methods,
but when the ground truth clustering only contains a few clusters,
it seems that \textsf{Ls-AD-VC}
can give substantially better results than other methods.

% In this example, $a = 0.1$ corresponds to having
% roughly 67,000 triplets,
% which is clearly too much if these are all to be collected from human annotators.
% Annotating $\approx 7000$ triplets
% is still a lot of work, but might be more realistic to achieve.
% This corresponds to $a = 0.01$,
% meaning that we could expect to be able to reconstruct an $f^*$
% having perhaps $k=4$ (or fewer) clusters
% using \textsf{Ls-AD-VC}.

Finally,
the rightmost column of Figure~\ref{fig:exp1main}
shows how the fraction of erroneous triplets $b$
affects the algorithms with fixed $k=8$ and different values of $a$.
We observe that when there are a lot of triplets ($a \geq 0.1$),
\textsf{Ls-AD-VC} is more or less unaffected by
the presence of noise,
and \textsf{t-STE} is a close second.
When there are only few triplets ($a = 0.01$),
\textsf{t-STE} works very well when there is no noise,
but its performance rapidly decreases as $b$ increases.
% Interestingly,
% as can be seen in the topmost panel,
% \textsf{CrowdClust} seems more sensitive to the presence of noise
% than the correlation clustering based methods
% even when there is an abundance of triplets ($a = 1$).

% There is further evidence to this
% in the right column of Figure~\ref{fig:exp1main},
% where se show the Rand index as a function of $k$
% for different fractions of noisy triplets.
% Also notice that this corresponds to a rather hard scenario
% where only a small fraction ($a = 0.01$) triplets are available.
% As soon as $b \geq 0.1$,
% we find that for all values of $k$,
% the local search algorithms slightly outperform \textsf{CrowdClust},
% with \textsf{Ls-AD-VC} having the best performance.
% However,
% as the number of clusters increases,
% all algorithms rather clearly fail to reconstruct $f^*$.

\begin{table*}[t]
\centering
\caption{Experiment 1: Average number of clusters found by the algorithms for different parameters of the triplet generator.}
\label{tbl:exp1k}
\begin{tabular}{ll|rrrr|rrrr|rrrr|rrrr}
\hline
\hline
 & & \multicolumn{4}{c|}{a = 1} & \multicolumn{4}{c|}{a = 0.1} & \multicolumn{4}{c|}{a = 0.05} & \multicolumn{4}{c}{a = 0.01} \\
 & & \multicolumn{4}{c|}{k} & \multicolumn{4}{c|}{k} & \multicolumn{4}{c|}{k} & \multicolumn{4}{c}{k} \\
 & & 2 & 4 & 8 & 16 & 2 & 4 & 8 & 16 & 2 & 4 & 8 & 16 & 2 & 4 & 8 & 16 \\
\hline
b = 0 & \textsf{Ls-AD-VC} & 2 & 4 & 8 & 16 & 2 & 4 & 8 & 14.3 & 2 & 4 & 8 & 10.4 & 2 & 3.7 & 4.1 & 3 \\
& \textsf{Ls-EQ-VC} & 2 & 3 & 4 & 6 & 2 & 2.8 & 2.8 & 2.4 & 2 & 2.7 & 2.5 & 2.4 & 2 & 2.5 & 2.2 & 2.2 \\
& \textsf{CrowdClust} & 6 & 4 & 6 & 6 & 3.5 & 4 & 4.8 & 3.8 & 3.1 & 4 & 4 & 2 & 2.1 & 3.4 & 3 & 1.4 \\
& \textsf{t-STE} & 4 & 7 & 7 & 14 & 4 & 4.9 & 7 & 13.4 & 4.2 & 4 & 7.1 & 13.7 & 2.5 & 4.3 & 7.5 & 11.9 \\
\hline
b = 0.1 & \textsf{Ls-AD-VC} & 2 & 4 & 8 & 16 & 2 & 4 & 7.9 & 12.7 & 2 & 4 & 7.9 & 6.5 & 2 & 3.8 & 3.5 & 2.8 \\
& \textsf{Ls-EQ-VC} & 2 & 3 & 4 & 5.3 & 2 & 2.6 & 2.9 & 2.2 & 2 & 2.2 & 2.4 & 2.3 & 2 & 2.4 & 2.3 & 2.3 \\
& \textsf{CrowdClust} & 3.7 & 4 & 5.2 & 5.6 & 2.8 & 3.6 & 5.2 & 1.8 & 2.4 & 3.7 & 2.6 & 1.7 & 2.6 & 2.5 & 1.2 & 1.2 \\
& \textsf{t-STE} & 1.4 & 7.4 & 7 & 11.6 & 1 & 5.3 & 7 & 12.4 & 1 & 4 & 7.3 & 10.2 & 2 & 4 & 5.9 & 3.9 \\
\hline
b = 0.2 & \textsf{Ls-AD-VC} & 2 & 4 & 8 & 16 & 2 & 4 & 7.7 & 12.4 & 2 & 4 & 6.9 & 3.5 & 2 & 3.4 & 2.8 & 2.4 \\
& \textsf{Ls-EQ-VC} & 2 & 3 & 4.3 & 4.6 & 2 & 2.6 & 2.5 & 2.4 & 2 & 2.5 & 2.1 & 2.4 & 2 & 2.4 & 2.4 & 2.3 \\
& \textsf{CrowdClust} & 2.8 & 4 & 3.3 & 5.3 & 2.5 & 3.8 & 3.9 & 1.3 & 1.8 & 3.9 & 3.4 & 1 & 1.7 & 2.1 & 1.2 & 1 \\
& \textsf{t-STE} & 1 & 8.2 & 6.8 & 10 & 1 & 5.1 & 6.7 & 5.4 & 1 & 3.3 & 6.3 & 3.4 & 1 & 1.9 & 1.2 & 1.1 \\
\hline
\hline
\end{tabular}

\end{table*}
What kind of errors do the algorithms make, then?
A simple way to address this question is
to consider the size of the resulting clustering.
Indeed, all algorithms should in theory be able to
find the correct number of clusters in ideal settings.
Table~\ref{tbl:exp1k} shows
the number of clusters (averages over 10 instances)
in the solution returned by the algorithms
for different values of $a$, $b$ and $k$.
(Here we only consider
\textsf{Ls-AD-VC}, \textsf{Ls-EQ-VC}, \textsf{CrowdClust}, and \textsf{t-STE}.)
We find that in most cases,
the algorithms tend to {\em underestimate the number of clusters}.
% The better performing methoods,
% \textsf{Ls-AD-VC} and \textsf{t-STE},
% almost always find a larger number of clusters.
As suggested by the results in Fig.~\ref{fig:exp1main},
\textsf{Ls-AD-VC} performs very well
when the number of triplets is large ($a$ is large),
irrespectively of the amount of noise in the input.
Also, it tends to outperform \textsf{t-STE}
for small values of $a$ when there is a lot of noise ($b$ increases)
and $k$ is small.
% Also, here benefits of using the \textbf{all different} initialisation for
% the local search algorithm become apparent,
% as \textsf{Ls-EQ-VC} performs substantially worse than
% \textsf{Ls-AD-VC} and \textsf{t-STE},
% possibly because the local search heuristic is unable
% to increase the number of clusters enough.

\subsection{Experiment 2: Example clusterings with real triplet data}
\begin{figure}
\centering
\includegraphics[width=\columnwidth]{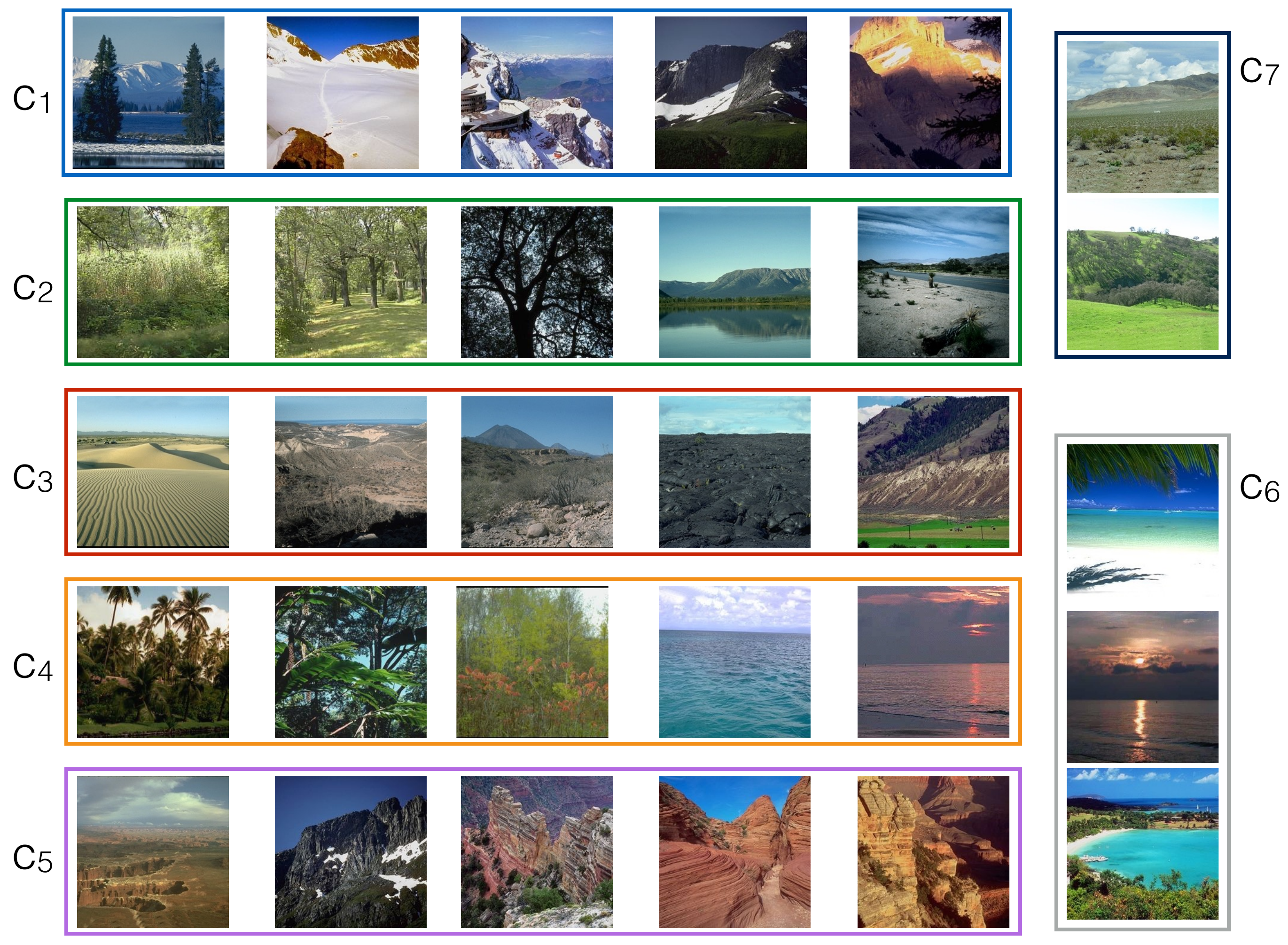}
\caption{Clustering of the \textbf{Nature} dataset found by \textsf{Ls-AD-VC}. The figure shows a {\em random sample of five items} from clusters 1--5, and clusters 6 and 7 in full. Cluster 1 contains images of snowy mountains, cluster 2 of wooden flat areas or forest, cluster 3 of drylands and deserts. Cluster 4 contains a mixture of tropical trees and open water with a clearly visible horizon, cluster 5 contains images of rocky mountains. Clusters 6 and 7 seem like extra clusters / outliers that could also be merged with some of the larger clusters.}
\label{fig:nature}
\end{figure}
\begin{figure}
\centering
\includegraphics[width=\columnwidth]{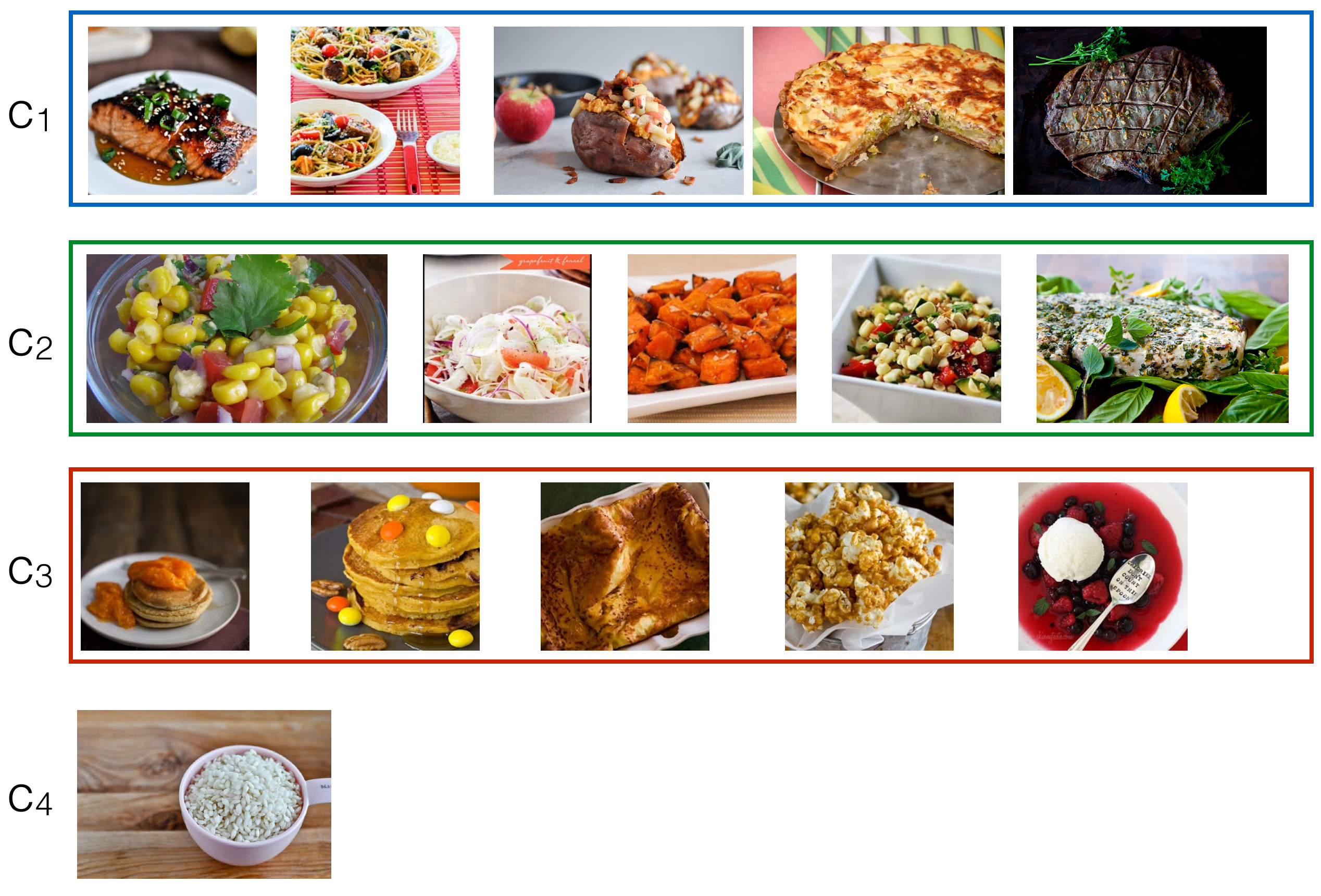}
\caption{Clustering of the \textbf{Food} dataset found by \textsf{Ls-AD-VC}. The figure shows a {\em random sample of five items} from clusters 1--3, cluster 4 only contains a single item. Cluster 1 contains images of savoury main courses, cluster 2 contains images of salads and other kinds of vegetables, cluster 3 corresponds to desserts and sweet dishes, and cluster 4 is a single outlier with an image of a cup of rice.}
\label{fig:food}
\end{figure}
\textbf{Setup:} In our second experiment
we present two case studies that
highlight how the \textsf{Ls-AD-VC} algorithm
works on real, crowdsourced data.
We obtained two sets of relative distance comparison triplets
from the authors of \cite{HeikinheimoU13} and \cite{WilberKB14}.
The first one (\textbf{Nature}),
originally used in \cite{HeikinheimoU13} to run a crowd-powered $k$-means algorithm,
contains a set of 3357 triplets
over 120 images of natural scenes
of four categories
({\em coast}, {\em open country}, {\em forest}, {\em mountain}),
from the Scene
image collection\footnote{\texttt{http://cvcl.mit.edu/database.htm}} \cite{OlivaT01}.
The second one (\textbf{Food}),
collected by the authors of \cite{WilberKB14} from Yummly.com,
contains 190,376 triplets over 100 images of
various dishes of food.

Note that the two datasets are of different ``densities''
(recall the parameter $a$ from above):
\textbf{Nature} contains only about 1 percent of all possible $120 \choose 3$ triplets,
while \textbf{Food} contains in fact {\em more than 100 percent}
of all possible triplets, i.e., \textbf{Food} contains several
instances of duplicated triplets.

We did {\em not perform any kind of pre-processing or cleanup
of the data in either case}.
The triplets may thus be noisy, conflicting
(i.e., for the same set of three items,
two triplets may indicate different items as the outlier), etc.
This decision was deliberate,
because we want to
illustrate the experience of a naive user,
who wants results quickly, and as simply as possible,
e.g.~without first running
a complex consensus model \cite{SheshadriL13} to clean up erroneous inputs.
Certainly there are situations in which
using those methods is really necessary,
but here we want to show
how \textsf{Ls-AD-VC} fares with
``raw'' human computation data.
(One can always argue that results
should only improve with more complex pre-processing.)
We thus run the experiment simply
by running \textsf{Ls-AD-VC} on all available triplets.

\textbf{Results:}
The found clusterings are shown in figures \ref{fig:nature} and \ref{fig:food}
for \textbf{Nature} and \textbf{Food}, respectively.
Of clusters with more than five items,
we show only {\em five randomly selected items}.
From \textbf{Nature} the method finds 7 clusters,
two of which are very small,
while from \textbf{Food} we find four clusters,
one of which contains only a single(!) item.
As can be seen from the figures,
the found clusterings are very intuitive in both cases.
In particular, thanks to the large number of triplets,
the result with \textbf{Food} is extremely good,
and really provides strong evidence for
the correlation clustering based method to be
both simple and very efficient.
Also, it is rather remarkable that the method
can find a very good clustering from \textbf{Nature},
despite there being very few triplets.
However, this is
in accordance with the results from synthetic clusterings in Experiment 1,
where we show that \textsf{Ls-AD-VC}
can find reasonable clusterings even
in this situation ($a = 0.01$)
as long as the underlying clustering is not too fine grained.

Finally,
we want to point out that
\textsf{Ls-AD-VC} is extremely fast:
the runtimes with \textbf{Nature} and \textbf{Food}
are $< 1$ second and $\approx 11$ seconds, respectively.
% Of course in human computation applications
% the main bottleneck is collecting the triplets,
% running times of the algorithms are less interesting to begin with.

\section{Conclusion and Future Work}
We defined, analysed,
and provided both an approximation algorithm,
as well as a practical local search algorithm for
a novel \textsc{Correlation-Clustering} variant based on relative distance comparisons.
We also showed empirically that
the approach has certain advantages over existing methods
for clustering with relative distance comparisons.

Our method is motivated by human computation approaches to clustering.
However,
an interesting property of our approach is that
it does not require guessing the number of clusters in advance.
Given that an arbitrary distance matrix can always be
used to generate relative distance comparisons in a parameter-free manner,
it seems interesting to
investigate if and how the approach could be used as a generic
non-parametric clustering algorithm.
Moreover,
it seems relevant to understand what kind of clusterings
our method can find when used in such ways.
For instance, does
the distance between two clusters and cluster diameter affect the outcome?
Finally,
implementing
agglomerative clustering approaches
using relative distances also seems of interest.
First steps towards this have been taken in \cite{HaghiriGL17}.

\bibliographystyle{IEEEtran}
\bibliography{rcc} 

% Generated by IEEEtran.bst, version: 1.12 (2007/01/11)
\begin{thebibliography}{10}
\providecommand{\url}[1]{#1}
\csname url@samestyle\endcsname
\providecommand{\newblock}{\relax}
\providecommand{\bibinfo}[2]{#2}
\providecommand{\BIBentrySTDinterwordspacing}{\spaceskip=0pt\relax}
\providecommand{\BIBentryALTinterwordstretchfactor}{4}
\providecommand{\BIBentryALTinterwordspacing}{\spaceskip=\fontdimen2\font plus
\BIBentryALTinterwordstretchfactor\fontdimen3\font minus
  \fontdimen4\font\relax}
\providecommand{\BIBforeignlanguage}[2]{{%
\expandafter\ifx\csname l@#1\endcsname\relax
\typeout{** WARNING: IEEEtran.bst: No hyphenation pattern has been}%
\typeout{** loaded for the language `#1'. Using the pattern for}%
\typeout{** the default language instead.}%
\else
\language=\csname l@#1\endcsname
\fi
#2}}
\providecommand{\BIBdecl}{\relax}
\BIBdecl

\bibitem{XuW05}
R.~Xu and D.~C.~W. II, ``Survey of clustering algorithms,'' \emph{{IEEE} Trans.
  Neural Networks}, vol.~16, no.~3, pp. 645--678, 2005.

\bibitem{law2011human}
E.~Law and L.~v. Ahn, ``Human computation,'' \emph{Synthesis Lectures on
  Artificial Intelligence and Machine Learning}, vol.~5, no.~3, pp. 1--121,
  2011.

\bibitem{QuinnB11}
A.~J. Quinn and B.~B. Bederson, ``Human computation: a survey and taxonomy of a
  growing field,'' in \emph{Proceedings of {CHI}}, 2011, pp. 1403--1412.

\bibitem{Ipeirotis10}
P.~G. Ipeirotis, ``Analyzing the amazon mechanical turk marketplace,''
  \emph{{ACM} Crossroads}, vol.~17, no.~2, pp. 16--21, 2010.

\bibitem{RussakovskyDSKS15}
O.~Russakovsky, J.~Deng, H.~Su, J.~Krause, S.~Satheesh, S.~Ma, Z.~Huang,
  A.~Karpathy, A.~Khosla, M.~S. Bernstein, A.~C. Berg, and F.~Li, ``Imagenet
  large scale visual recognition challenge,'' \emph{International Journal of
  Computer Vision}, vol. 115, no.~3, pp. 211--252, 2015.

\bibitem{lintott2008galaxy}
C.~J. Lintott, K.~Schawinski, A.~Slosar, K.~Land, S.~Bamford, D.~Thomas, M.~J.
  Raddick, R.~C. Nichol, A.~Szalay, D.~Andreescu \emph{et~al.}, ``Galaxy zoo:
  morphologies derived from visual inspection of galaxies from the sloan
  digital sky survey,'' \emph{Monthly Notices of the Royal Astronomical
  Society}, vol. 389, no.~3, pp. 1179--1189, 2008.

\bibitem{mohammad2013crowdsourcing}
S.~M. Mohammad and P.~D. Turney, ``Crowdsourcing a word--emotion association
  lexicon,'' \emph{Computational Intelligence}, vol.~29, no.~3, pp. 436--465,
  2013.

\bibitem{PeiFTR16}
Y.~Pei, X.~Z. Fern, T.~V. Tjahja, and R.~Rosales, ``Comparing clustering with
  pairwise and relative constraints: {A} unified framework,'' \emph{{TKDD}},
  vol.~11, no.~2, pp. 22:1--22:26, 2016.

\bibitem{HeikinheimoU13}
H.~Heikinheimo and A.~Ukkonen, ``The crowd-median algorithm,'' in
  \emph{Proceedings of {HCOMP}}, 2013.

\bibitem{UkkonenDH15}
A.~Ukkonen, B.~Derakhshan, and H.~Heikinheimo, ``Crowdsourced nonparametric
  density estimation using relative distances,'' in \emph{Proceedings of
  {HCOMP}}, 2015, pp. 188--197.

\bibitem{KleindessnerL16}
M.~Kleindessner and U.~von Luxburg, ``Lens depth function and k-relative
  neighborhood graph: versatile tools for ordinal data analysis,'' \emph{CoRR},
  vol. abs/1602.07194, 2016.

\bibitem{SchultzJ03}
M.~Schultz and T.~Joachims, ``Learning a distance metric from relative
  comparisons,'' in \emph{Proceedings of {NIPS}}, 2003, pp. 41--48.

\bibitem{GomesWKP11}
R.~Gomes, P.~Welinder, A.~Krause, and P.~Perona, ``Crowdclustering,'' in
  \emph{Proceedings of {NIPS}}, 2011, pp. 558--566.

\bibitem{TamuzLBSK11}
O.~Tamuz, C.~Liu, S.~J. Belongie, O.~Shamir, and A.~Kalai, ``Adaptively
  learning the crowd kernel,'' in \emph{Proceedings of {ICML}}, 2011, pp.
  673--680.

\bibitem{AmidGU15}
E.~Amid, A.~Gionis, and A.~Ukkonen, ``A kernel-learning approach to
  semi-supervised clustering with relative distance comparisons,'' in
  \emph{Proceedings of {ECML} {PKDD}}, 2015, pp. 219--234.

\bibitem{KleindessnerL16a}
M.~Kleindessner and U.~von Luxburg, ``Kernel functions based on triplet
  similarity comparisons,'' \emph{CoRR}, vol. abs/1607.08456, 2016.

\bibitem{MaatenW12}
L.~van~der Maaten and K.~Q. Weinberger, ``Stochastic triplet embedding,'' in
  \emph{Proceedings of {IEEE} {MLSP}}, 2012, pp. 1--6.

\bibitem{AmidU15}
E.~Amid and A.~Ukkonen, ``Multiview triplet embedding: Learning attributes in
  multiple maps,'' in \emph{Proceedings of {ICML}}, 2015, pp. 1472--1480.

\bibitem{AmidVW16}
E.~Amid, N.~Vlassis, and M.~K. Warmuth, ``t-exponential triplet embedding,''
  \emph{CoRR}, vol. abs/1611.09957, 2016.

\bibitem{BansalBC04}
N.~Bansal, A.~Blum, and S.~Chawla, ``Correlation clustering,'' \emph{Machine
  Learning}, vol.~56, no. 1-3, pp. 89--113, 2004.

\bibitem{BohmKKZ04}
C.~B{\"{o}}hm, K.~Kailing, P.~Kr{\"{o}}ger, and A.~Zimek, ``Computing clusters
  of correlation connected objects,'' in \emph{Proceedings of {ACM} {SIGMOD}},
  2004, pp. 455--466.

\bibitem{DemaineEFI06}
E.~D. Demaine, D.~Emanuel, A.~Fiat, and N.~Immorlica, ``Correlation clustering
  in general weighted graphs,'' \emph{Theor. Comput. Sci.}, vol. 361, no. 2-3,
  pp. 172--187, 2006.

\bibitem{AilonCN08}
N.~Ailon, M.~Charikar, and A.~Newman, ``Aggregating inconsistent information:
  Ranking and clustering,'' \emph{J. {ACM}}, vol.~55, no.~5, pp. 23:1--23:27,
  2008.

\bibitem{MonienS85}
B.~Monien and E.~Speckenmeyer, ``Ramsey numbers and an approximation algorithm
  for the vertex cover problem,'' \emph{Acta Inf.}, vol.~22, no.~1, pp.
  115--123, 1985.

\bibitem{PapadimitriouS82}
C.~H. Papadimitriou and K.~Steiglitz, \emph{Combinatorial Optimization:
  Algorithms and Complexity}.\hskip 1em plus 0.5em minus 0.4em\relax
  Prentice-Hall, 1982.

\bibitem{CormenLRS01}
T.~H. Cormen, C.~E. Leiserson, R.~L. Rivest, and C.~Stein, \emph{Introduction
  to Algorithms, Second Edition}.\hskip 1em plus 0.5em minus 0.4em\relax The
  {MIT} Press and McGraw-Hill Book Company, 2001.

\bibitem{GomesWP08}
R.~Gomes, M.~Welling, and P.~Perona, ``Incremental learning of nonparametric
  bayesian mixture models,'' in \emph{Proceedings of {IEEE} {CVPR}}, 2008.

\bibitem{hubert1985comparing}
L.~Hubert and P.~Arabie, ``Comparing partitions,'' \emph{Journal of
  classification}, vol.~2, no.~1, pp. 193--218, 1985.

\bibitem{WilberKB14}
M.~J. Wilber, I.~S. Kwak, and S.~J. Belongie, ``Cost-effective hits for
  relative similarity comparisons,'' in \emph{Proceedings of {HCOMP}}, 2014.

\bibitem{OlivaT01}
A.~Oliva and A.~Torralba, ``Modeling the shape of the scene: {A} holistic
  representation of the spatial envelope,'' \emph{International Journal of
  Computer Vision}, vol.~42, no.~3, pp. 145--175, 2001.

\bibitem{SheshadriL13}
A.~Sheshadri and M.~Lease, ``{SQUARE:} {A} benchmark for research on computing
  crowd consensus,'' in \emph{Proceedings of {HCOMP}}, 2013.

\bibitem{HaghiriGL17}
S.~Haghiri, D.~Ghoshdastidar, and U.~von Luxburg, ``Comparison based nearest
  neighbor search,'' \emph{CoRR}, vol. abs/1704.01460, 2017.

\end{thebibliography}

\appendix
\subsection{Proof of Lemma~\ref{lemma:inconsistency}}
\label{app:inconsistency}
\begin{proof}
For two triplets $t_1$ and $t_2$ to be inconsistent,
there must exist the items $a$ and $b$ such that
$t_1$ is satisfied only when $f(a) = f(b)$,
and $t_2$ is satisfied only when $f(a) \neq f(b)$.
% Consider all those triplets in $\triplets^\prime$
% that are inconsistent due to items $a$ and $b$.
Consider only those edges in $C_\triplets$ that exist
due to inconsistencies
induced by the items $a$ and $b$.
These edges must form a bipartite clique.
%% Indeed, all pairs of items induce a bipartite clique in $C$.
% The minimum vertex cover of $C$
% must also cover the edges of all of these bipartitie cliques.
Since $\vc(C_\triplets)$ is a vertex cover of $C_\triplets$,
it {\em must} contain a vertex cover of
every bipartite clique in $C$.
A proper vertex cover of any bipartite clique must
contain at least all vertices from ``one side'' of the clique,
otherwise it cannot cover all edges in the clique.
By the above reasoning,
in every bipartitie clique of $C$
at least one of the ``sides'' must be completely contained in $\vc(C)$.
Since all triplets that belong to
the same ``side'' of a bipartite clique
agree about the items $a$ and $b$,
the triplets that remain in $\triplets \setminus \vc(C)$
can not be inconsistent.
\end{proof}

\subsection{Proof of Theorem~\ref{thm:approx}}
\label{app:approxproof}
Before giving the proof, we present three technical lemmas that are needed
to establish the result.
\begin{lemma}
\label{lemma:sw}
Let $\triplets^\prime = \triplets \setminus \vc(C_\triplets)$,
and denote by $G_{\triplets^\prime}$ the associated \textsc{Correlation-Clustering} instance.
We have
\[
s(f, \triplets^\prime) \leq c(f, G_{\triplets^\prime}) \leq 2 \; s(f, \triplets^\prime)
\]
for any clustering function $f$,
where $c$ and $s$ are the cost functions of
Problems \ref{prob:cc} and \ref{prob:rcc}, respectively.
\end{lemma}
\begin{proof}
We rewrite the cost function $c$ of Problem~\ref{prob:cc}
as a sum over $\triplets^\prime$.
By definition of $G_{\triplets^\prime}$,
the weight $w(u,v)$ is equal to the
size of $\triplets^\prime_{u,v}$, i.e.,
the number of triplets
that contain both items $u$ and $v$.
We can write
\[
\begin{split}
c(f, G_{\triplets^\prime}) = \sum_{(u,v) \in E^+} & \sum_{t \in \triplets^\prime_{u,v}} \I\{ f(u) \neq f(v) \} + \\
&  \sum_{(u,v) \in E^-} \sum_{t \in \triplets^\prime_{u,v}} \I\{ f(u) = f(v) \}.
\end{split}
\]
Instead of summing over $E^+$ and $E^-$ separately,
we simply sum over all $(u,v)$ and use an indicator function to select the appropriate case.
Note that any pair $(u,v)$ can only belong to either $E^+$ or $E^-$ but not both.
This yields
\[
\begin{split}
c(f, G_{\triplets^\prime}) = \sum_{(u,v)} & \sum_{t \in \triplets^\prime_{u,v}} ( \I\{(u,v) \in E^+\} \I\{ f(u) \neq f(v) \} + \\
&  \I\{(u,v) \in E^-\} \I\{ f(u) = f(v) \} ).
\end{split}
\]
We then change the order of summation to run first over $\triplets^\prime$, and then over the pairs in every triplet:
\[
\begin{split}
c(f, G_{\triplets^\prime}) = \sum_{t \in \triplets^\prime} & \sum_{(u,v) \in t} ( \I\{(u,v) \in E^+\} \I\{ f(u) \neq f(v) \} + \\
&  \I\{(u,v) \in E^-\} \I\{ f(u) = f(v) \} ).
\end{split}
\]
Finally,
because in every triplet $t = (a,b,c)$ there is precisely one pair, $(a,b)$, that belongs to $E^+$, while the other two pairs, $(a,c)$ and $(b,c)$, belong to $E^-$, we obtain:
\[
\begin{split}
c(f, G_{\triplets^\prime}) = & \sum_{(a,b,c) \in \triplets^\prime} ( \I\{ f(a) \neq f(b) \} \; + \\
& \I\{ f(a) = f(c) \} \; + \; \I\{ f(b) = f(c) \} ).
\end{split}
\]
This has the same form as the cost $s$ of Problem~\ref{prob:rcc},
with the difference that in $c(f, G_{\triplets^\prime})$
the cost of every triplet $(a,b,c) \in \triplets^\prime$
is the sum of three indicator variables,
instead of a single indicator with a disjunction of three conditions,
as is the case in $s(f, \triplets^\prime)$.
The conditions, however, are the same in both cases.
First, this implies that $s(f, \triplets^\prime) \leq c(f, G_{\triplets^\prime})$
for every $f$.
Second, it is easy to see that
for a given triplet $(a,b,c)$,
at most {\em two} of these conditions can be satisfied simultaneously.
Whenever the triplet $(a,b,c)$ incurs a cost of $1$ in $s(f,\triplets^\prime)$,
it therefore incurs at most a cost of $2$ in $c(f, G_{\triplets^\prime})$
for every $f$.
And if none of the conditions is satisfied,
the triplet $(a,b,c)$ incurs a cost of zero in both cases.
This implies the 2nd inequality of the Lemma.
\end{proof}

\begin{lemma}
\label{lemma:decomp}
We have
$s(f, \triplets) \leq |\vc(C_\triplets)| + s(f, \triplets^\prime)$
for any vertex cover $\vc(C_\triplets)$ and clustering function $f$.
\end{lemma}
\begin{proof}
By definition, $\vc(C_\triplets)$ and $\triplets^\prime$ constitute
a disjoint partition of $\triplets$.
Since the cost function $s$ is simply a sum over triplets in $\triplets$,
we have $s(f, \triplets) = s(f, \vc(C_\triplets)) + s(f, \triplets^\prime)$
for any $f$.
Also, clearly $s(f, \vc(C_\triplets))$ can be at most
$|\vc(C_\triplets)|$ (all triplets in $\vc(C_\triplets)$ are violated),
which leads to the inequality of the Lemma.
\end{proof}

\begin{lemma}
\label{lemma:vcbound}
Let $\triplets$ be a set of triplets,
$C_\triplets$ denote the associated constraint graph,
$\vc^{\min}(C_\triplets)$ the minimum vertex cover of $C_\triplets$,
and let $f^\opt$ denote the optimal solution to Problem~\ref{prob:rcc}.
We have
$|\vc^{\min}(C_\triplets)| \leq s(f^\opt, \triplets)$.
\end{lemma}
\begin{proof}
Every edge in $C_\triplets$ consists of a pair of triplets,
of which at least one {\em must be unsatisfied}
for {\em any} clustering $f$, including $f^\opt$.
The minimum vertex cover $\vc^{\min}(C_\triplets)$ corresponds to
the smallest possible subset of triplets
that must be unsatisfied,
resulting in a lower bound of $s(f^\opt, \triplets)$.
\end{proof}

We conclude with the proof of Theorem~\ref{thm:approx}:
\begin{proof}
The proof is a mechanical exercise with the lemmas presented above.
We start from Lemma~\ref{lemma:decomp} that holds for any $f$ and $\vc(C_\triplets)$:
\begin{eqnarray*}
&& s( f^{\alg}, \triplets ) \leq |\vc^{\alg}(C_\triplets)| + s( f^{\alg}, \triplets^\prime ) \hspace{2em} \\
&\leq& |\vc^{\alg}(C_\triplets)| + c( f^{\alg}, G_{\triplets^\prime} ) \hspace{2em}  \hbox{(Lemma~\ref{lemma:sw})} \\
&\leq& \beta |\vc^{\min}(C_\triplets)| + \alpha c( f^\opt_{CC}, G_{\triplets^\prime} ) \hspace{1em}  \hbox{(Assum.~\ref{asm:ccapp} and \ref{asm:vcapp})} \\
&\leq& \beta |\vc^{\min}(C_\triplets)| + \alpha c( f^\opt, G_{\triplets^\prime} ) \hspace{1em}  \hbox{($f^\opt_{CC}$ is optimal)}\\
&\leq& \beta |\vc^{\min}(C_\triplets)| + 2 \alpha s( f^\opt, \triplets^\prime ) \hspace{1em}  \hbox{(Lemma~\ref{lemma:sw})}\\
&\leq& \beta |\vc^{\min}(C_\triplets)| + 2 \alpha s( f^\opt, \triplets ) \hspace{2em}  \hbox{($\triplets^\prime \subseteq \triplets$)}\\
&\leq& \beta s( f^\opt, \triplets ) + 2 \alpha s( f^\opt, \triplets ) \hspace{2em}  \hbox{(Lemma~\ref{lemma:vcbound})} \\
&=& (2\alpha + \beta) \; s( f^\opt, \triplets ).
\end{eqnarray*}
\end{proof}

\end{document}